\documentclass[letterpaper,12pt]{article}
\makeatletter
\def\ps@pprintTitle{%
 \let\@oddhead\@empty
 \let\@evenhead\@empty
 \def\@oddfoot{\centerline{\thepage}}%
 \let\@evenfoot\@oddfoot}
\makeatother

\newcommand{\Var}{\mbox{var}}

\newcommand{\iid}{\stackrel{\mathrm{iid}}{\sim}}
\newcommand{\ind}{\stackrel{\mathrm{ind}}{\sim}}

\newcommand{\argmin}{\arg\!\min}

\newcommand{\mbf}[1]{\mathbf{#1}}

\usepackage{amsthm}
\usepackage{amsmath}
\SetSymbolFont{letters}{bold}{OML}{cmm}{b}{it}
\SetSymbolFont{operators}{bold}{OT1}{cmr}{bx}{n}
\usepackage{amssymb, amsfonts}
\usepackage{natbib}
\usepackage{bm}
\usepackage{graphicx}
\usepackage{color}
\usepackage{floatrow}
\usepackage{units}
\usepackage{mathabx}
\usepackage{url}
\usepackage{verbatim}
\usepackage{geometry}
\usepackage{slashbox}
\usepackage{epstopdf}
\usepackage{booktabs,caption,fixltx2e}
\usepackage[flushleft]{threeparttable}
\usepackage{rotating}
\usepackage{multirow}
\usepackage{physics}
\usepackage{hyperref}
\usepackage{nameref}
\usepackage{mathptmx} 
\usepackage{subfig} 
\usepackage{placeins} 

\newtheorem{theorem}{Theorem}[section]
\newtheorem{lemma}[theorem]{Lemma}

\makeatletter
\let\orgdescriptionlabel\descriptionlabel
\renewcommand*{\descriptionlabel}[1]{%
  \let\orglabel\label
  \let\label\@gobble
  \phantomsection
  \edef\@currentlabel{#1}%
  \let\label\orglabel
  \orgdescriptionlabel{#1}%
}
\makeatother

\begin{document}

\title{Scalable Bayes under Informative Sampling}
\date{\today}
\author{Terrance D. Savitsky\thanks{U.S. Bureau of Labor Statistics, 2 Massachusetts Ave. N.E, Washington, D.C. 20212 USA}, Sanvesh Srivastava\thanks{Department of Statistics and Actuarial Science, The University of Iowa, Iowa City, Iowa, USA}}

\maketitle

\begin{abstract}
  Bayesian hierarchical formulations are utilized by the U.S. Bureau of Labor Statistics (BLS) with respondent-level data for missing item imputation because these formulations are readily parameterized to capture correlation structures. BLS collects survey data under informative sampling designs that assign probabilities of inclusion to be correlated with the response on which sampling-weighted pseudo posterior distributions are estimated for asymptotically unbiased inference about population model parameters. Computation is expensive and does not support BLS production schedules.  We propose a new method to scale the computation that divides the data into smaller subsets, estimates a sampling-weighted pseudo posterior distribution, in parallel, for every subset, and combines the pseudo posterior parameter samples from all the subsets through their mean in the Wasserstein space of order 2. We construct conditions on a class of sampling designs where posterior consistency of the proposed method is achieved. We demonstrate on both synthetic data and in application to the Current Employment Statistics survey that our method produces results of similar accuracy as the usual approach, while offering substantially faster computation.
\end{abstract}


\section{Introduction}
Bayesian hierarchical models are popular for inference and imputation in complex data because latent dispositional states that underlie observed behaviors and induce a correlation structure can be directly parameterized \citep{RSSC:RSSC12049}. Bayesian models readily support multiple imputation of missing data in a fashion that captures uncertainty in estimation of model parameters (under a missing at random assumption) \citep{citeulike:12855856}.

These models are extensively employed for estimations on data acquired from surveys.  Survey data are often collected using informative sampling designs that induce a correlation between inclusion probabilities assigned to units in a target population and the response variable of interest.  Our inferential interest is the joint estimation of population model parameters and imputation of missing items for respondent-level data acquired under an informative sampling design.

Existing Bayesian methods applied to data acquired under informative sampling designs, however, focus on design-based inference for domain-indexed (e.g., area) summary statistics, rather than estimation of population model parameters \citep{dong:2014,kunihama:2014,wu:2010,si2015}, which is our focus.  The recent approach of \citet{2015arXiv150707050S} formulates a sampling-weighted pseudo posterior distribution to approximate the population posterior distribution of interest, while preserving the posterior sampling geometry for parameters of any model specified by the data analyst. The pseudo posterior computations often become intractable, however, due to the relatively large size of the observed respondent-level data. Motivated by this problem, we propose a new method based on the divide-and-conquer technique that extends the application of \citet{2015arXiv150707050S} to large-sized respondent-level data by scaling the computation, while preserving the useful property of minimal modification to the analyst-specified model or posterior sampling geometry. Our extension generalizes the \emph{Wasserstein Posterior} approach for scalable Bayesian inference due to \citet{Srietal15,2015arXiv150805880S} to account for informative sampling designs. This extension is extremely efficient and supports the rapid turnaround cycles used by BLS to publish the employment statistics on a monthly basis.

Our method consists of three steps. Firstly, the sampled units are randomly split into disjoint subsets such that computation in each subset is tractable. Secondly, we construct a sampling-weighted pseudo posterior distribution for model parameters that is estimated in each data subset of the observed sampled units.  We normalize the sampling weights used to formulate the pseudo posterior in each subset to sum to the total number of observations in the observed sample in order to scale the variance of each subset posterior distribution to match that of the observed full sample data.  Thirdly, we combine these pseudo posterior distributions by computing their barycenter in their Wasserstein space of order 2. The computation of a barycenter from subset pseudo posteriors scales sublinearly in sample size because each subset pseudo posterior estimation may be run in parallel, limited only by computational resources. The proposed method is applicable to the same class of sampling designs as outlined in \citet{2015arXiv150707050S}.  We demonstrate theoretical results that show if the number of subsets are chosen appropriately, then our Generalized Wasserstein pseudo posterior (GWPP) method, applied under the class of informative sampling designs specified in \citet{2015arXiv150707050S}, converges to the true parameter at a near optimal rate.

\section{Motivating Data: Current Employment Statistics Survey}\label{application}
The United States Bureau of Labor Statistics (BLS) administers the Current Employment Statistics survey (CES) to non-farm, public, and private business establishments across the United States on a monthly basis, receiving approximately 270,000 submitted responses in each month, or over $3$ million responses in a year. Estimated total employment is published for detailed industry categories by state and for selected metropolitan areas.  The survey uses a stratified sampling design with strata constructed by combinations of state, broad industry grouping, and employment size divided into $8$ categories. The business establishments are sampled by their unique unemployment insurance tax identification numbers, which may contain a cluster of multiple individual sites. {If a business establishment is selected based on its unique identification number, then all of the associated sites in that cluster are also included.}  Stratum-indexed inclusion probabilities are set to be proportional to the average employment size for member establishments of that stratum.

The CES constructs a \emph{known} sampling design distribution that assigns higher inclusion probabilities to establishments with a relatively larger number of employees. This is a proportion-to-size design that induces a correlation among sample inclusion probabilities and total employment; larger establishments more strongly influence the variance of domain-indexed total employment statistics published by BLS. Such sampling designs are called \emph{informative} because they induce a correlation between selection probabilities and observed values. In this survey, distributions of establishment employment counts for samples will be skewed to higher values than present in the underlying population. If the informativeness in the design is not modeled, then inference on population parameters conditional on the survey data will be biased \citep{2015arXiv150707050S}.

There is a short time gap between the receipt of establishment submissions at the end of a month and the subsequent publication of employment estimates for that month; the joint estimation of population model parameters and imputation for missing items, followed by the computation of employment statistics for reported domains must be performed quickly. The relatively large number of submissions with non-zero changes in employment levels, coupled with the rapid publication schedule, require the use of computationally scalable estimation tools. The sampling-weighted pseudo posterior distribution proposed in \citet{2015arXiv150707050S} fails to meet these requirements, motivating our development of the GWPP for computationally efficient estimation of (population) model parameters and imputation of missing responses on data acquired under an informative sampling design.

Section~\ref{gsa} introduces the pseudo posterior distribution as it will be applied by the data analyst, in practice.
The likelihood used in the pseudo posterior distribution is constructed, in practice, by exponentiating the likelihood contribution for each unit by its associated sampling weight, which is inversely proportional to the unit inclusion probability.  We review the Wasserstein space of measures and computation for the barycenter of subset distributions in this section.  We highlight the concept of stochastic approximation that exponentiates each subset likelihood contribution by a scale factor such that the subset posterior distribution provides a noisy quantification of the uncertainty in the posterior distribution for the full data.  We leverage \citet{2015arXiv150707050S} in Section~\ref{sa} to generalize stochastic approximation to construct subset pseudo posterior distributions for data acquired under an informative sampling design.  Section~\ref{theory} provides theoretical conditions on the sampling design which guarantee the
in-expectation contraction of any subset posterior distribution and in-probability contraction of the GWPP to a delta measure centered at the true parameter value under the Wasserstein metric of order 2. We apply the GWPP for inference in a multivariate employment count response model for synthetic data and for the data acquired from the CES survey in Section~\ref{sec:data-analysis}. We demonstrate that the GWPP and the posterior distribution computed using full data are close in total variation distance.  We conclude this paper with a discussion in Section~\ref{discussion}.

\section{Generalizing Stochastic Approximation}\label{gsa}
\subsection{Preliminaries: Wasserstein Barycenter}\label{wasserstein}
The order $2$ Wasserstein space probability measures are defined on a separable and complete metric space, $\left(\Theta,\rho\right)$.  Let $\Pi_1, \ldots, \Pi_K$ be $K$ probability measures defined on this space of probability measures.  \citet{2015arXiv150805880S} introduces an associated order 2 Wasserstein metric $(W_2)$, which permits the computation of a \emph{barycenter}, $\overline{\Pi}$, of the $K$ probability measures, defined as that probability measure which minimizes the sum of squared $W_2$ distances to the $K$ probability measures.  A more formal introduction is performed in the next, theoretical exposition section.

The Wasserstein barycenter motivates the \emph{Wasserstein Posterior} approach for scalable Bayesian inference \citep{Srietal15}.  Let $y_1, \ldots, y_N$ be data for units in a finite population, $U$, of size $\lvert U\rvert = N$.  Without loss of generality, suppose we divide the units into $K$ equally-sized subsets, $\{U_{j}\}_{j=1,\ldots,K}$, of equal size $\lvert U_{j}\rvert = M$, such that $N=KM$ and subset $j$ includes data $y_{[j]} = \{y_{j1},\ldots,y_{jM}\}$ ($j=1, \ldots, K$). Further, suppose $\Pi_j(\cdot \mid {y}_{[j]}) $ and $\Pi(\cdot \mid y_1, \ldots, y_N)$ are posterior distributions for $\theta \in \Theta$ conditioned on subset $j$ and full data, respectively.   The \emph{Wasserstein Posterior}, denoted as $\overline{\Pi}\left(\cdot \mid y_1, \ldots, y_N \right)$, is the Wasserstein barycenter of $\Pi_{j}\left(\cdot \mid {y}_{[j]}\right)$ ($j=1, \ldots, K$). If posterior draws are available from $\Pi_{j}\left(\cdot \mid {y}_{[j]}\right)$ ($j=1, \ldots, K$), then an empirical approximation of $\overline{\Pi}\left(\cdot \mid y_1, \ldots, y_N \right)$ can be estimated by solving a linear program using those draws; see \citet{Srietal15} for details.

\citet{2015arXiv150805880S} construct a noisy posterior approximation for the population from subset $j$ with,
\begin{equation}
\pi\left(\theta\vert y_{[j]}\right) \propto \left[\mathop{\prod}_{i = 1}^{M}p\left(y_{ji}\vert \theta\right)^{\gamma}\right]\pi\left(\theta\right), \label{wasppseudopost}
\end{equation}
where $\gamma = N/M = K$ exponentiates each likelihood contribution so that uncertainty quantification from the subset of size $M$ approximates that of size $N$, which produces a barycenter estimate, $\overline{\Pi}$, whose estimated posterior variance is of the same order as the posterior distribution estimated on the full data.  \citet{2015arXiv150805880S} refer to this exponentiation of the likelihood contributions for uncertainty quantification as ``stochastic approximation".

\subsection{Preliminaries: Pseudo Posterior Distribution}\label{pseudo}
Under random sampling of the finite population, we don't observe the full population, $U$, but a sample taken from it, $S \subset U$, where $\lvert S\rvert = n \leq N$.  Let $\delta_{i} \in \{0,1\}$ denote the sample inclusion indicator for units $i = 1,\ldots,N$ from the population.  The density for the observed sample is denoted by, $\pi\left(\mathbf{y}_{o}\vert\theta\right) = \pi\left(\mathbf{y}\vert \{\delta_{i} = 1\}_{i=1,\ldots,N},\theta\right)$, where ``$o$" indicates ``observed".

\citet{2015arXiv150707050S} define a pseudo posterior distribution tuned for the theoretical setup of informative sampling. They construct a plug-in approximation for the finite population posterior density estimated on the observed sample as
\begin{equation}
  \pi^{\pi}\left(\theta \mid y_{o,1}, \ldots,y_{o,n}, \tilde w_1, \ldots, \tilde w_n \right) \propto \left\{ \mathop{\prod}_{i = 1}^{n} p\left(y_{o,i} \mid \theta \right)^{\tilde{w}_{i}}  \right\} \pi \left(\theta\right), \label{pseudopost}
\end{equation}
where $\pi (\theta)$ is the prior parameter density, $\tilde{w}_{i} = n w_i (\sum_{i=1}^n w_{i})^{-1}$ ($i = 1,\ldots,n$), with $w_{i} = 1/\pi_{i}$ for $\pi_i$ defined is the marginal inclusion probability of unit $i$. The exponent $\tilde{w}_{i}$ corrects for sampling informativeness and ensures that  $\tilde{w}_{i}$ assigns the relative importance of the likelihood contribution of unit $i$ to approximate the likelihood for the population.  The scaling factor here is $1$ in that weights are scaled to the sample size, $n$, which asymptotically expresses the amount of information present in our observed sample.

The sampled observations are often dependent in design distributions under the informative sampling.
\citet{2015arXiv150707050S} define a condition under which the sampling design distribution produces samples which are asymptotically independent as the finite population size, $N$, increases, which is needed to guarantee $L_{1}$ contraction.  In practice, many sampling designs obey this condition, including the design for the Current Employment Statistics survey, where the number of establishments increases within each industry and state in the limit.  There are two additional conditions that restrict the class of sampling designs required for consistency and they are formally reviewed in Section~\ref{theory}.  We will drop the subscript $``o"$ in $y_o$ in the sequel because our focus is on data acquired from a sample of a finite population.

\subsection{Generalized Stochastic Approximation}\label{sa}
In many applications sampling from the pseudo density in \eqref{pseudopost} is computationally expensive and it is easier to sample from a pseudo posterior density conditioned on a data \emph{subset}. The observed sample, $S \subset U$, (henceforth referred to as the ``full sample", which is not to be confused with the ``full data" associated to the population, $U$) is first divided into disjoint $K$ disjoint subsets, $S_j$ $(j=1, \ldots, K)$, each of equal size, $m = \lvert S_{j} \rvert = n/K$ (where equal size is chosen for ease-of-exposition without loss of generality) such that $\displaystyle S = S_1 \cup \cdots \cup S_K$.  We construct a pseudo likelihood for density, $p\left(y_{ji}\mid \theta\right)$, for unit $i \in S_{j}$, by exponentiating it with its sample weight, $\tilde{w}_{ji}$, to form,
\begin{equation}
  \pi^{\pi}\left(\theta \mid y_{[j]}\right) \propto \left( \mathop{\prod}_{i = 1}^{m}p\left(y_{ji}\vert\theta\right)^{\tilde{w}_{ji}}\right)\pi\left(\theta\right)\label{genpseudopost}
\end{equation}
We redefine $\tilde{w}_{ji}$ as $n w_i (\sum_{i \in S_j} w_{i})^{-1}$ ($j = 1,\ldots,K$), that normalizes the weights in each subset to sum to $n$, the full sample size, rather than $m$, such that variance of $\theta$ with density $\pi^{\pi}\left(\theta \mid y_{[j]}\right)$ ($j = 1,\ldots,K$) is of the same order as that of $\pi^{\pi}\left(\theta \mid y_{o,1}, \ldots,y_{o,n}, \tilde w_1, \ldots, \tilde w_n \right)$ in \eqref{pseudopost}. This ensures that all subset pseudo posterior distribution are noisy approximations of the full sample pseudo posterior distribution.

The GWPP is computed as the barycenter of $K$ subset pseudo posterior distributions with densities defined in \eqref{genpseudopost}. It provides an approximation to the partially-observed finite population posterior density under informative sampling. We next outline the theoretical properties of the GWPP computed using $K$ subset pseudo posterior distributions, which are each scaled by the vector of sampling weights.

\section{Consistency of Generalized Wasserstein Pseudo Posterior}\label{theory}

\subsection{Setup}

Consider the theoretical setup for an informative sampling design. Let $\nu$ be a positive integer, and $U_{\nu}$ is a finite population of size $\vert U_{\nu} \vert = N_{\nu}$ such that if  $\nu < \nu^{'}$, then $N_{\nu} < N_{\nu^{'}}$. Under our setup $\{N_{\nu}\}_{\nu\in\mathbb{N}}$ is an increasing sequence of population sizes, with $\lim_{\nu\uparrow\infty}N_{\nu} = \infty$.  Let $Y_{\nu 1},\ldots, Y_{\nu N_{\nu}}$ be a sequence of independent and \emph{non}-identically distributed (\emph{inid}) random variables that are defined for the $N_{\nu}$ units in population $U_{\nu}$ and take values on the measurable product space $\mathop{\otimes}_{i=1}^{N_{\nu}}\left(\mathcal{Y}_{\nu i},\mathcal{A}_{\nu i}\right)$, where $\mathcal{A}_{\nu i}$ is the Borel sigma-algebra on $\mathcal{Y}_{\nu i}$ ($i=1, \ldots, N_{\nu}$). The asymptotics under our construction is controlled by $\nu\in\mathbb{N}$ to map to the process where we fix a $\nu$, construct an associated finite population of size, $N_{\nu}$, generate random variables $Y_{\nu 1},\ldots, Y_{\nu N_{\nu}} \sim P_{\theta_{0}}$, construct unit marginal sample inclusion probabilities, $(\pi_{\nu 1},\ldots,\pi_{\nu N}$ under $P_{\nu}$ and then draw a sample, $\{1,\ldots,n_{\nu}\}$ from that population.  The process is repeated for each increment of $\nu$ such that the entire vector of response variable values and unit inclusion probabilities are regenerated. See \citet{bonnery:2013} for a recent theoretical exposition of model consistency under informative sampling that indexes a sequence of populations by $\nu$.

For any parameter $\theta  \in \Theta \subset \mathbb{R}^p$, let $P_{\theta \nu i}$ represent the probability distribution of $Y_{\nu i}$ indexed by $\theta$ that has the density $dP_{\theta \nu i} = p\left(y_{\nu i} \mid  \theta\right) d \theta$ relative to a sigma-finite measure $\mu_{i}$ ($i = 1,\ldots,N_{\nu}$). Define the product measure $P^{N_{\nu}}_{\theta}$ on $\mathop{\otimes}_{i=1}^{N_{\nu}}\left(\mathcal{Y}_{\nu i},\mathcal{A}_{\nu i}\right)$ as $P^{N_{\nu}}_{\theta} = \mathop{\otimes}_{i=1}^{N_{\nu}}P_{\theta i}$ that has density $\prod_{i=1}^{N_{\nu}} p\left(y_{\nu i} \mid  \theta\right)$ with respect to $\mathop{\otimes}_{i=1}^{N_{\nu}}\mu_{i}$. We write $Y_{\nu i}$, $y_{\nu i}$, and $P^{N_{\nu}}_{\theta}$ as $Y_{i}$, $y_{i}$, and $P_{\theta}$ for brevity in the remainder of the paper because the context is clear.

\subsection{Pseudo Posterior Distribution}

The observed data are sampled from the finite population, $U_{\nu}$, under a survey sampling design that induces a \emph{known} distribution, $P_{\nu}$, defined on a vector of random inclusion indicators for the population units, $\bm{\delta}_{\nu} = \left(\delta_{\nu 1},\ldots,\delta_{\nu N_{\nu}}\right)$, where $\delta_{\nu i} \in \left\{0,1\right\}$ indexes inclusion of unit $i$ in observed sample, $S_{\nu}$.  The joint distribution over $\left(\delta_{\nu 1},\ldots,\delta_{\nu N_{\nu}}\right)$ is described by known marginal unit inclusion probabilities, $\pi_{\nu i} = \mbox{Pr}\left\{\delta_{\nu i} = 1\right\}$ for all $i \in U_{\nu}$ and the second-order pairwise probabilities, $\pi_{\nu i\ell} = \mbox{Pr}\left\{\delta_{\nu i} = 1 \cap \delta_{\nu \ell} = 1\right\}$ for $i,\ell \in U_{\nu}$.

In the sequel, we further divide the $n_{\nu}$ observed units into $K$ disjoint subsets that, for ease-of-exposition, we suppose are all of size $m_{\nu} < n_{\nu}$.  We conduct parallel model estimations on each sample subset (of size $m_{\nu}$) such that each provides a noisy approximation to the posterior distribution estimated on the full sample.   Without loss of generality for exposition of our consistency results that directly follow, we suppose a collection of populations, $\left\{U_{\nu j}\right\}_{j = 1,\ldots,K}$, each of size, $M_{\nu} < N_{\nu}$, that exhaust $U_{\nu} = \mathop{\bigcup}_{j =1,\ldots,K} U_{\nu j}$. The $K$ populations are all generated from, $P_{\theta}$, with density, $p\left(y_{ji}\middle\vert \theta\right)$. We subsequently take a sample from each $U_{\nu j}$ under $P_{\nu}$, the sampling design distribution.  The resulting set of $K$ samples are typically \emph{dependent} due to the without replacement sampling design where, fixing a $j \in \left\{1,\ldots,K\right\}$, the inclusion probability of a unit in $U_{\nu j}$ will depend on whether units in $\left\{U_{\nu \ell}\right\}_{\ell \neq j \in (1,\ldots,K)}$ are co-included.  The two steps of drawing a sample (of observed data) from the finite population and subsequent division into disjoint subsets are re-cast as a single (informative without replacement) sampling step from the collection of $K$ disjoint finite populations.  We extend notations, $\pi_{\nu ji} = \mbox{Pr}\left\{\delta_{\nu ji} = 1\right\}$ and $\pi_{\nu ji\ell} = \mbox{Pr}\left\{\delta_{\nu ji} = 1 \cap \delta_{\nu j\ell} = 1\right\}$ for $i,\ell \in U_{\nu j}$.

Our task is to perform inference about the unknown true, $\theta_{0}$, that we suppose generates the finite population from $P_{\theta_{0}}$, by assigning a prior measure $\Pi$ with density $\pi$ on the parameter space $\Theta$ such that $\theta_0 \in \Theta$.  We construct a sampling-weighted pseudo likelihood as in \citet{2015arXiv150707050S} by defining
\begin{align}
  p^{\pi}_{\theta ji} = p\left(y_{ji} \mid \theta \right)^{\frac{\delta_{\nu ji}}{\pi_{\nu ji}}},\quad i \in U_{\nu j}. \label{eq:wtss}
\end{align}
The likelihood contribution of sample $i$ in subset $j$ is weighted by $\pi_{\nu ji}^{-1}$ in \eqref{eq:wtss} so that the information in subset $j$ approximates the information in partially observed finite population of size $M_{\nu}$. We use the pseudo likelihood in \eqref{eq:wtss} and the prior $\pi(\theta)$ to obtain the pseudo posterior density for subset $j$ as
\begin{equation}\label{inform_post}
  \pi^{\pi}_{j}\left(\theta  \mid y_{[j]} \delta_{\nu [j]} \right) = \frac{\mathop{\prod}_{i\in[j]}\frac{p^{\pi}_{\theta ji}}{p^{\pi}_{\theta_{0} ji}}\pi(\theta)}{\mathop{\int}_{\Theta}\mathop{\prod}_{i\in [j]}\frac{p^{\pi}_{\theta ji}}{p^{\pi}_{\theta_{0} ji}}\pi(\theta)d\theta},
\end{equation}
where  $[j] = \left\{i \in U_{\nu j}\right\}$ denotes the $M_{\nu}$ finite population units in $U_{\nu j}$, $y_{[j]} = \{y_{ji} : i \in U_{\nu j}\}$, and $\delta_{\nu [j]}  = \{\delta_{\nu ji} : i \in U_{\nu j}\}$. The sampling weights $\pi_{\nu ji}$ $(i \in S_{\nu j})$ in the observed sub-sample, $S_{\nu j} \subseteq U_{\nu j}$,  satisfy $ \mathop{\sum}_{i\in S_{\nu j}} \pi^{-1}_{\nu ji} = n_{\nu}$ so that $\pi^{\pi}_{j}\left(\theta  \mid y_{[j]} \delta_{\nu [j]} \right)$ ($j = 1,\ldots,K$) is a noisy approximation of the posterior density defined on the  observed sample of size $n_{\nu}$, $\pi\left(\theta \mid \{y_i : \delta_{\nu i} = 1, i = 1, \ldots, N_{\nu}\}\right)$. We recover the subset pseudo posterior density defined in \citet{2015arXiv150805880S} if we set $ {\delta}_{\nu [j]} = \left(1,\ldots,1\right)$ in \eqref{inform_post}.

\subsection{Generalized Wasserstein Pseudo Posterior Distribution}

We construct the GWPP to combine $K$ subset pseudo posterior distributions estimated using \eqref{inform_post}. Let $\Pi^{\pi}_{j}\left(\cdot  \mid y_{[j]} \delta_{\nu [j]} \right)$ ($j=1, \ldots, K$) represent the $K$ subset posterior posteriors and $\overline{\Pi}^{{\pi}}(\cdot \mid \{y_i : \delta_{\nu i} = 1, i = 1, \ldots, N_{\nu}\})$ represent the GWPP. The event probabilities in the informative sampling designs are denoted by $P_{\theta_{0},P_{\nu}}$, which is indexed by  $\theta_{0}$ and $P_{\nu}$ to indicate the joint distribution with respect to generation of the finite population and subsequent taking of the observed sample. The resulting sample observations taken from $U_{\nu j}$ under $P_{\nu}$ are now dependent due to the dependence induced by sampling without replacement. We extend the definition of the Wasserstein space of probability measures, $\mathcal{P}_{2}\left(\Theta\right)$, from \citet{2015arXiv150805880S} to define
\begin{equation*}
  \mathcal{P}_{2 \nu} \left(\Theta\right) = \bigg\{\mu_{\nu} :  \mathop{\int}_{\theta\in\Theta} \rho\left(\theta_{0},\theta\right)^{2} \mu_{\nu} \left(d\theta\right) < \infty\bigg\}.
\end{equation*}
Assuming $\Pi^{\pi}_{j}\left(\cdot  \mid y_{[j]} \delta_{\nu [j]} \right) \in \mathcal{P}_{2 \nu}\left(\Theta\right) $ $(j=1, \ldots,K)$, we extend the definition of the associated barycenter from \citet{2015arXiv150805880S} to define the generalized Wasserstein pseudo posterior as
\begin{equation*}
  \overline{\Pi}^{{\pi}} = \mathop{\argmin}_{\Pi \in \mathcal{P}_{2 \nu}\left(\Theta\right)} \, \frac{1}{K} \mathop{\sum}_{j=1}^{K} W_{2}^{2} \left(  \Pi,\Pi_{j}^{\pi}\right),
\end{equation*}
and Proposition 3.8 in \citet{AguCar11} implies that $\overline{\Pi}^{{\pi}}$ exists uniquely in $ \mathcal{P}_{2 \nu} \left(\Theta\right)$.  Our employment of subscript, $\nu$, accounts for the dependence of the resulting pseudo posterior distribution of \eqref{inform_post} on the sampling design distribution, $P_{\nu}$.

\subsection{Empirical process functionals}
We will approximate the joint distribution for population generation and informative sampling using an empirical distribution construction similar to \citet{breslow:2007} that incorporates inverse inclusion probability weights, $1 / \pi_{\nu ji}$ $(i=1, \ldots M_{\nu})$,
\begin{equation}
{P}^{\pi}_{M_{\nu}} = \frac{1}{M_{v}}\mathop{\sum}_{i=1}^{M{\nu}}\frac{\delta_{\nu ji}}{\pi_{\nu ji}}\delta\left(Y_{ji}\right),
\end{equation}
where $\delta\left(Y_{ji}\right)$ denotes the Dirac delta function, with probability mass $1$ on observed $Y_{ji}$ and we recall that $M_{\nu} = \vert U_{\nu j} \vert$ denotes the size of of the finite population for subset $j$.

We follow the notational convention of \citet{Ghosal00convergencerates} and define the associated expectation functionals with respect to these empirical distributions by ${P}^{\pi}_{M_{\nu}}f = \frac{1}{M_{\nu}}\mathop{\sum}_{i=1}^{M_{\nu}}\frac{\delta_{\nu ji}}{\pi_{\nu ji}}f\left(Y_{ji}\right)$.  Similarly, ${P}_{M_{\nu}}f = \frac{1}{M_{\nu}}\mathop{\sum}_{i=1}^{M_{\nu}}f\left(Y_{ji}\right)$ for $f:\mathcal{Y}\rightarrow {\mathbb{R}}$.  Associated centered empirical processes are defined, ${G}^{\pi}_{M_{\nu}} = \sqrt{M_{\nu}}\left({P}^{\pi}_{M_{\nu}}-P_{0}\right)$ and ${G}_{M_{\nu}} = \sqrt{M_{\nu}}\left({P}_{M_{\nu}}-P_{0}\right)$.

The sampling-weighted, pseudo Hellinger distance between densities defined on $\theta_{1},\theta_{2} \in \Theta$, $h^{\pi,2}_{M_{\nu}}\left(\theta_{1},\theta_{2}\right) := \left[h^{\pi}_{M_{\nu}}\left(\theta_{1},\theta_{2}\right)\right]^{2} = \frac{1}{M_{\nu}}\mathop{\sum}_{i=1}^{M_{\nu}}\frac{\delta_{\nu ji}}{\pi_{\nu ji}}h^{2}\left(p_{\theta_{1},ji},p_{\theta_{2},ji}\right)$, where $h\left(p_{1},p_{2}\right) = \left\{ \mathop{\int}\left(\sqrt{p_{1}}-\sqrt{p_{2}}\right)^{2}d\mu \right\}^{\frac{1}{2}}$ for dominating measure, $\mu$.  The associated non-sampling Hellinger distance is specified with, $h^{2}_{M_{\nu}}\left(\theta_{1},\theta_{2}\right) = \frac{1}{M_{\nu}}\mathop{\sum}_{i=1}^{M_{\nu}}h^{2}\left(p_{\theta_{1},ji},p_{\theta_{2},ji}\right)$.  We later assume that $h^{\pi}_{M_{\nu}}\left(\theta,\theta_{0}\right)$ is lower bounded by a constant multiple of $\rho\left(\theta, \theta_0\right)$. This assumption is used in deriving the rate of contraction of the subset pseudo posterior distributions to a delta measure centered on $\theta_0$ ($\delta_{\theta_0}$) in $W_2$ metric.

\subsection{Main Results}
We next specify six conditions for the metric space, $\left(\Theta,\rho\right)$, and the associated prior on the space, $\Pi$, followed by the three additional conditions on the sampling design distribution, $P_{\nu}$. Suppose we have a  sequence, $\epsilon_{M_{\nu}} \downarrow 0$ and $M_{\nu}\epsilon^{2}_{M_{\nu}}\uparrow\infty$ as positive integer $\nu \uparrow\infty$,
\begin{description}
\item[(A1)\label{bounded}] (Non-zero inclusion probabilities) Define constant $\displaystyle\gamma \geq 1: \mathop{\sup}_{\nu} \left( \mathop{\max}_{i\in U_{\nu j}}\frac{1}{\pi_{\nu ji}} \right) \leq \gamma,  \text{ for all } j = 1,\ldots,K$, uniformly, and constants $g_{1},g_{2} > 0$ where $g_{1}\gamma M_{\nu} \leq N_{\nu} \leq g_{2}\gamma M_{\nu}$.
\item[(A2)\label{independence}] (Asymptotic Independence Condition)
     \begin{equation*}
        \displaystyle\mathop{\limsup}_{\nu\uparrow\infty} \mathop{\max}_{i \neq \ell\in U_{\nu j}}\left\vert\frac{\pi_{\nu ji\ell}}{\pi_{\nu ji}\pi_{\nu j\ell}} - 1\right\vert = O(M_{\nu}^{-1}) \text{  with $P_{\theta_0}$-probability $1$}
     \end{equation*}
     such that for some constant, $c_{3} > 0$, and sufficiently large $M_{\nu}$,
     \newline$\displaystyle M_{\nu}\mathop{\sup}_{\nu}\mathop{\max}_{i \neq \ell \in U_{\nu j}}\left[\frac{\pi_{\nu ji\ell}}{\pi_{\nu ji}\pi_{\nu j\ell}}-1\right] \leq c_{3} , \text{ for all } j = 1,\ldots,K$, uniformly.
\item[(A3)\label{compact}] (Compactness) $\Theta$ is a compact space in the $\rho$ metric and $\theta_{0}$ is an interior point of $\Theta$.
\item[(A4)\label{distance}] (Pseudo Distance bounded from below) For any $\theta_{1},\theta_{2} \in \Theta$ and $j = 1,\ldots,K$, there exists a positive constant, $C_{L}$, such that:
    \begin{equation*}
    \mathop{\min}_{\delta_{\nu j} \in\Delta_{\nu j}}h^{2}_{M_{\nu}}\left(\theta_{1},\theta_{2}\right)  \geq C_{L}\rho^{2}\left(\theta_{1},\theta_{2}\right),
    \end{equation*}
    where $\delta_{\nu j} = (\delta_{\nu j1}\in\{0,1\},\ldots,\delta_{\nu j M_{\nu}})$ denotes a selected sample (of size $m_{\nu}$), drawn from the space of all possible samples, $\Delta_{\nu j}$, such that
    $\displaystyle\mathop{\sum}_{\delta_{\nu j}\in\Delta_{\nu j}}\mbox{Pr}_{P_{\nu}}\left(\delta_{\nu j}\right) = 1$,
\item[(A5)\label{size}] (Local entropy condition - Size of model)
    Let constants $D_{1} > 0$ and $0 < D_{2} < \frac{D_{1}^{2}}{2^{12}\gamma^{2}}$, and define a function, $\Phi\left(u,r\right) \geq 0$, increasing in $u \in \mathbb{R}^{+}$, non-decreasing in $r\in \mathbb{R}^{+}$,  such that for all sufficiently large $M_{\nu}$,
        \begin{equation*}
        H_{[]}\left(u,\left\{\theta\in\Theta: h_{M_{\nu}}\left(\theta,\theta_{0}\right) \leq r\right\},h_{M_{\nu} }\right) \leq \Phi\left(u,r\right),
        \end{equation*}
        where $H_{[]}$ denotes the $h_{M_{\nu}}$-bracketing entropy, which is the $\log$ of $1 +$ the bracketing number defined for data drawn independently in \citet{2015arXiv150805880S}, and the size of the bracketing entropy bound is restricted to,
       \begin{equation*}
        \mathop{\int}_{D_{1}\frac{r^{2}}{12}}^{D_{1}r}\sqrt{\Phi\left(u,r\right)}du < D_{2}\sqrt{M_{\nu}}r^{2}.
        \end{equation*}
\item[(A6)\label{thickness}] (Prior thickness) There exist positive constants, $\kappa$ and $c_{\pi}$ such that uniformly over all $j = 1,\ldots,K$,
    \begin{equation*}
    \Pi \left\{ \theta\in\Theta:\frac{1}{M_{\nu}}\mathop{\sum}_{i=1}^{M_{\nu}} \mathbb{E}_{P_{\theta_{0}}}\exp\left(\kappa\log_{+} \frac{p_{\theta_{0} ji}}{p_{\theta ji}}  \right)-1 \leq\epsilon_{M_{\nu}}^{2} \right\} \geq \exp\left(-c_{\pi}\kappa M_{\nu}\epsilon_{M_{\nu}}^{2}\right),
    \end{equation*}
    where $\log_{+}x = \max(\log x,0)$, for $x > 0$.
\item[(A7)\label{convexity}] (Convexity of metric) The metric, $\rho$, satisfies that for any positive integer $N_{\nu}$, $\theta_{1},\ldots,\theta_{N},\theta^{'} \in \Theta$ and non-negative weights, $\mathop{\sum}_{i=1}^{N_{\nu}} w_i = 1$,
    \begin{equation*}
    \rho\left(\mathop{\sum}_{i=1}^{N_{\nu}}w_{i}\theta_{i},\theta^{'}\right) \leq \mathop{\sum}_{i=1}^{N_{\nu}}w_{i}\rho\left(\theta_{i},\theta^{'}\right).
    \end{equation*}
\end{description}

A few comments about our assumptions are in order. Assumptions \nameref{bounded} and \nameref{independence} are the same as those used in \citet{2015arXiv150707050S} and, together, impose conditions on the sampling distribution, $P_{\nu}$, that define a restricted class of sampling designs.  Assumption \nameref{bounded} requires the sampling design to assign a positive probability for inclusion of every unit in the finite population. No portion of the population may be systematically excluded, which would prevent a sample of any size from containing information about the population from which the sample is taken.  Assumption \nameref{independence} restricts the result to sampling designs where the dependence among lowest-level sampled units attenuates to $0$ as $\nu\uparrow\infty$; for example, a two-stage sampling design of clusters within strata would meet this condition if the number of population units nested within each cluster from which the sample is drawn increases in the limit of $\nu$.  Multi-stage sampling designs of individuals within households, which are in turn, nested within geographically-indexed primary sampling units (PSUs) would appear to violate this requirement for asymptotic independence of unit inclusions because the number of individuals within each household remains fixed in the limit of $\nu$; however, it is our experience based on upcoming research that the within household dependence of individuals is overwhelmed by the relative independence between households and PSUs, such that the marginally-weighted pseudo posterior distribution effectively meets this condition.

Assumptions \nameref{compact} -- \nameref{convexity} follow from \citet{2015arXiv150805880S}.  Theorem~\ref{main_subset} will show the contraction of the subset pseudo posterior distributions to $\delta_{\theta_{0}}$ under $W_2$ metric in expectation, $\mathbb{E}_{P_{\theta_{0}}, P_{\nu}}$.  In Assumption \nameref{distance}, the value of $h^{\pi}_{M_{\nu}}$ for a realized sample, $\delta_{\nu j}$, of size, $m_{\nu}$, drawn from a subset population of size, $M_{\nu}$, is a noisy approximation of $h_{M_{\nu}}$ defined on the whole population, since the contribution from each unit, $\ell \in 1,\ldots,m_{\nu}$, used to construct $h^{\pi}_{M_{\nu}}$, is upweighted (by its inverse inclusion probability) to represent its concentration in the population.  Assumption \nameref{size} alters the assumption that regulates model complexity from \citet{2015arXiv150805880S} by inserting $\gamma^{2}$ in the denominator of the upper limit for $D_{2}$, which restricts the bracketing entropy.   Sampling designs with larger $\gamma$ will, on average, produce samples whose information expresses more variation about that of the population, so that the allowed size of the model space under which consistency is guaranteed declines as $\gamma$ increases.  Assumption \nameref{thickness} imposes a stronger exponential decay control over the tail probability than the condition that averages $L_{2}$ norms of the log-likelihood ratio evaluated at the finite population data values specified in Theorem 4 of \citet{ghosal2007}; however, we still use these assumptions for easy comparisons between the results in this work and in \citet{2015arXiv150805880S}. There is no loss of generality as the result goes through with the condition from Theorem 4 of \citet{ghosal2007} with minor modifications.

Our first result guarantees that if our assumptions hold, then each subset-indexed pseudo posterior distribution contracts to a delta measure centered on the true model generating parameters under $W_2$ metric in expectation, $\mathbb{E}_{P_{\theta_{0}}, P_{\nu}}$. This notion of contraction is stronger than the commonly studied contraction rate in $(P_{\theta_{0}}, P_{\nu})$-probability.

\begin{theorem}
\label{main_subset}
Suppose assumptions \nameref{bounded} -- \nameref{convexity} hold for subset pseudo posteriors, $\Pi^{\pi}_{j}\left(\cdot  \mid y_{[j]}\delta_{\nu[j]}\right)$ ($j=1, \ldots, K$). Then there exist positive constants $c_{1}, r_{1}, r_{2}, \gamma, c_{4}, $ and large constant, $B_{0} = \max_{\theta\in\Theta}\rho\left(\theta,\theta_{0}\right)$, such that for sufficiently large $M_{\nu}$,
\begin{align}
\mathbb{E}_{P_{\theta_{0}},P_{\nu}} \left[ W_{2}^{2} \left\{  \Pi_{j}^{\pi}\left(\cdot \mid y_{[j]} {\delta}_{\nu [j]}\right),\delta_{\theta_{0}}(\cdot) \right\} \right] \leq & c_{1}^{2}\epsilon_{M_{\nu}}^{2} + \nonumber \\
& B_{0}\left[\frac{1}{r_{2}M_{\nu}\epsilon_{M_{\nu}}^{2}} + 5\exp\left(-r_{1}c_{4} N_{\nu}\epsilon_{M_{\nu}}^{2}\right)\right],
\end{align}
uniformly for all $j = 1,\ldots,K$, where $r_{1}\geq \frac{\left(c_{\pi}g_{2} + 3\left(\kappa\gamma\right)^{-1}\right)}{g_{1}},~r_{2}= \frac{1}{\left[c_{3} + 1 + \gamma\right]} \leq 1,~c_{1} = \sqrt{\frac{2 r_{1} g_{2}\gamma^{2}}{q_{1} C_{L} }}, c_{4} = \min\left(\frac{q_{2}}{q_{1}},1\right)$.
\end{theorem}

We note that the rate of convergence is injured for a sampling distribution, $P_{\nu}$, that assigns relatively low inclusion probabilities to some units in the finite population such that $\gamma$ will be relatively larger. Constants $r_{1}$ and $r_{2}$ decrease, while $c_{1}$ increases as $\gamma$ becomes larger. Samples drawn under a design that induces a large variability in the sampling weights will express more dispersion in their information similarity to the underlying finite population, and so will contract on the truth at a relatively slower rate.  Similarly, the larger the dependence among the finite population unit inclusions induced by $P_{\nu}$, the higher will be $c_{3}$ and the slower will be the rate of contraction.  While our consistency result focuses on contraction of the sampling-weighted pseudo posterior distribution onto the true generating parameters, rather than the true posterior distribution, results in
\citet{2015arXiv150707050S} demonstrate that the pseudo posterior distribution contracts onto the true posterior distribution, in practice.  They compare the pseudo posterior distribution estimated on an informative sample to the posterior distribution estimated on an equally-weighted, simple random sample, with both samples taken from the same population.  The pseudo posterior distribution quickly (as sample size increases) removes bias and ensures robust coverage of the $95\%$ credible interval.  The relative variance (and coverage lengths) of the pseudo posterior distribution may take relatively longer to contract to that of the posterior distribution, to the extent that the sampling design is less efficient than simple random sampling.

The source of bias from estimation of an unweighted (population) posterior distribution on observed data taken under an informative sample is the correlation between the unit inclusion probabilities and the response variable(s) of interest.  To the extent that sample inclusion probabilities (and, therefore, sampling weights) express variance unrelated to the response variables, the resulting pseudo posterior distribution will express relatively more variance than the posterior distribution (estimated on a simple random sample), without providing any bias correction.  It is therefore common to calibrate the weights to known population totals for one or more variables, which are fully observed for the whole population, to remove such excess variability, which would have the effect of lowering $\gamma$.  Estimated non-response weights, which are multiplied by the sampling weights to form a set of unit indexed total weights, would be expected to more quickly remove bias in the case where the non-response mechanism is correlated with the response variable(s).  Our method may be used without modification on published sampling weights that include nonresponse adjustments and a calibration step.

Our next result guarantees that if our assumptions hold, then the GWPP contracts to a delta measure centered on the true model generating parameters under $W_2$ metric in $(P_{\theta_0}, P_{\nu})$-probability.

\begin{theorem}
\label{main_wass}
Suppose conditions \nameref{bounded} -- \nameref{convexity} hold for subset pseudo posteriors, $\Pi^{\pi}_{j}\left(\cdot  \mid y_{[j]}\delta_{\nu[j]}\right)$ ($j=1, \ldots, K$). Then as $M_{\nu} \uparrow \infty$ under fixed integer number of subsets, $K$,
\begin{align}
W_{2}\left\{ \overline{\Pi}^{{\pi}}(\cdot \mid \{y_i : \delta_{\nu i} = 1, i = 1, \ldots, N_{\nu}\}),\delta_{\theta_{0}}\left(\cdot\right)  \right\} =  O_{P} \left(\epsilon_{M_{\nu}}\right),
\end{align}
where $O_{P}$ is in $\left(P_{\theta_{0}},P_{\nu}\right)$-probability.
\end{theorem}

In practice, one may try to plug-in a value for $\epsilon_{M_{\nu}}$ that satisfies the conditions, $\epsilon_{M_{\nu}} \downarrow 0$ and $M_{\nu}\epsilon^{2}_{M_{\nu}}\uparrow\infty$ as the positive integer $\nu \uparrow\infty$, to the bound in Theorem~\ref{main_subset} and the convergence order in Theorem~\ref{main_wass} to see if the resultant bound limits to $0$; for example, choosing
$\epsilon_{M_{\nu}} = \left(\log^{2}M_{\nu}/M_{\nu}\right)^{1/2}$, used by \citet{2015arXiv150805880S} for so-called regular models;  for example, models with continuous densities, which are the class of models we specify in our Assumption (A2) works in both Theorems.

An important implication of our two results is that the data analyst may choose the number of subsets, $K$, based on their computational budget and expect that the resulting estimated GWPP will estimate arbitrarily closely to the {full sample} pseudo posterior distribution (for a moderate total sample size), but with a large savings in computation time. We demonstrate this performance in the sequel by estimating both the full sample pseudo posterior and the GWPP on our CES application.  Confidence in the GWPP is important in Federal statistical estimation as it will be impractical or impossible to estimate the model parameters using the full data.

These two theorems extend similar results of \citet{2015arXiv150805880S} for independent data to dependent data, where dependence is induced through the sampling design distribution, $P_{\nu}$; for example, sampling without replacement designs induce dependencies among units.  The proofs of both theorems generally follow from the techniques in  \citet{2015arXiv150805880S} with substantial modifications to account for informative sampling and the sampling design-induced dependence among the observations. Our approaches include two unique enabling lemmas and four additional lemmas that extend \citet{2015arXiv150805880S} to informative sampling . Proofs of the two theorems are in the Appendix and the proofs of enabling lemmas are in Section 1 of the Supplementary Material.

\section{Data Analysis}
\label{sec:data-analysis}
\subsection{Hierarchical Model for Current Employment Statistics Survey Data}\label{model}
Our motivating data consists of survey responses in the state of California in a $12$ month period from October, $2010$ to September, $2011$. Let $c$ index an establishment-by-month case observation for establishment $i$ and in month $t\{i\}$ ($i = 1, \ldots, n$; $t\{i\} = 1,\ldots,T_{i}$; $n_{\mbox{\tiny{cases}}} = \sum_{i=1}^n T_i$; $c = 1, \ldots, n_{\mbox{\tiny{cases}}}$). Let $T = \displaystyle\max\left(T_{1},\ldots,T_{n}\right)$ denote the number of unique months observed in the data. Let $\ell$  $(\ell = 1,\ldots,L)$ index the number of industries. We define industries using the North American Industry Classification System, which assigns a $6$-digit code over $1100$ industries.  We use the first two digits that denote the industry ``super-sectors'' for our data. There are $L = 23$ super-sectors populated by $n =  36390$ establishments in California; see Table 1 in the Supplementary Material for the definition of the super-sectors and the allocation of establishments.

The goal for our modeling is to use the temporal- and industry-indexed dependence among establishments to efficiently perform simultaneous estimation of population model parameters and imputation of missing values for one or more employment count variables. Noting that CES employment count variables, \emph{total number of employees} (ae), and the \emph{total number of production workers} (pw) (generally defined as non-supervisory workers) reported in the survey are highly dependent, we define a $Q$-dimensional response including these count variables, where $Q=2$. The number of missing values for the total number of employees in the survey is only $45$ out of $294674$. This is much smaller than the number of missing responses for the total number of production workers, which equals $142999$ out of $294674$. Accounting for the dependence between the total number of employees and the total number of production workers leads to better estimations of model parameters and imputation of missing responses than the case where dependence between the two responses is ignored.

We next construct a negative binomial sampling-weighted \emph{pseudo} likelihood for the observed sample of establishment employment counts from our survey data with,
\begin{align}
  y_{cq} \mid \tau_q, \, \psi_{cq} &\ind \mbox{NB} \left\{ \tau_{q},\exp(\psi_{cq}) \right\}^{\tilde{w}_{i\{c\}}}, \quad (c=1, \ldots, n_{\mbox{\tiny{cases}}}; \, q = 1, \ldots Q) \nonumber\\
  \mathop{{\psi}_{c}}^{Q\times 1} &= \left(\psi_{c1},\ldots,\psi_{cQ}\right)^{'} = {\theta}_{t\{c\}} + {\gamma}_{\ell\{c\}t\{c\}}z_{c}, \label{psieff}
\end{align}
{where $\ind$ denotes ``independently sampled from,'' ${\tilde{w}_{i\{c\}}}$ is the scaled sampling weight for establishment $i$ linked to case $c$, and NB represents the negative binomial distribution with  $\tau_q$ and $\exp(\psi_{cq})$ as its size and mean parameters}.  The indexing of precision parameters, $\tau_{q}$ (${q = 1,\ldots,Q}$), by employment count response variable, $q$, permits the by-variable modeling of over-dispersion present in each employment count variable from our data due to the large variation in the sizes of establishments in both the population and sample.  The $Q\times 1$ mean on the logarithm scale, ${\psi}_{c}$, is constructed from multivariate fixed and random effects.  The subscripts, $t\{c\},\ell\{c\},$ and $i\{c\}$, used to construct the mean on the logarithm scale in \eqref{psieff} denote the month $t$, industry $\ell$, and establishment $i$ linked to case observation $c$ $(t=1, \ldots, T$; $\ell=1, \ldots, L$; $ i=1, \ldots, n $; $c=1, \ldots, n_{\mbox{\tiny{cases}}})$. Fixed effect intercepts are denoted by the $Q \times T$ matrix, $\displaystyle\mathop{\Theta} = \left(\mathop{{\theta}_{1}},\ldots,{\theta}_{T}\right)$, indexed by response variable and month. We specify industry indexed  $Q \times T \times L$ random effects array, $\displaystyle\mathop{\Gamma} = \left(\mathop{\Gamma_{1}},\ldots,\Gamma_{L}\right)$, where the $Q\times 1$ vector ${\gamma}_{\ell t}$ models an effect for industry $\ell$ in month $t$, $\left(\ell = 1, \ldots, L;\, t = 1,\ldots,T\right)$.  We include industry-indexed random effects because we expect a dependence in the employment counts, $(y_{cq})$, over the months of interest for those establishments linked to the same industry (super sector).   Random effects predictor, $z_{c}$, represents the total employment for establishment, $i\{c\}$, on a $6$ month lagged basis in month, $t\{c\}$, obtained from a census instrument, the Quarterly Census of Employment and Wages.  The $6$ month lag derives from the relatively rapid Current Employment Statistics production schedule under which employment statistics are published on a more timely basis for this survey instrument than is possible for the Quarterly Census of Employment and Wages.  The historical values the Quarterly Census of Employment and Wages serve as a magnitude variable.  The two terms of \eqref{psieff} allow for non-linear associations over industries and months to each response variable.

We complete the specification of our probability model with the following priors,
\begin{subequations}
\label{dpmix}
\begin{align}
\mathop{\Theta}^{Q\times T} &\sim \mathcal{N}_{Q \times T}({0},\mathop{P_{2}^{-1}}^{Q\times Q} \circ \mathop{P_{3}^{-1}}^{T\times T}), \quad
\mathop{\Gamma_{\ell}}^{Q\times T} \iid \mathcal{N}_{Q \times T}({0},P_{8}^{-1}\circ P_{6}^{-1}) ~ (\ell = 1,\ldots,L),\label{priorGamma}\\
P_{s} &\sim \text{Huang-Wand}(\nu, b_{s1}, \ldots, b_{sQ}), ~b_{sq} \iid \mathcal{G}\left(1/2,1\right),~q=1,\ldots,Q ~ (s = 2, 8),\label{wish}\\
P_{s} &= D-r_{s}\Omega;~ r_{s} \sim \mathcal{U}(0,1)~ (s = 3, 6), \quad \tau_{q}^{-1/2} \iid \mathcal{C}(0,1),\label{car}
\end{align}
\end{subequations}
where $\iid$ denotes ``independently and indentically distributed as,'' Huang-Wand is a marginally noninformative prior for covariance matrices \citep{huang2013}, and $\mathcal{N}$, $\mathcal{U}$, $\mathcal{G}$, and $\mathcal{C}$ denote the Gaussian, uniform, Gamma and Cauchy distributions and $\circ$ denotes a tensor or outer product under a separable covariance specification of a matrix variate Gaussian (which is equivalent to employing a Kronecker product if $\Theta$ and $\Gamma_{\ell}$ were vectorized). The matrix $\Omega$ is a $T\times T$ adjacency matrix where $\omega_{ij} = 1$ if months $i$ and $j$ are adjacent; else, $\omega_{ij} = 0$, and $D$ is a $T\times T$ diagonal matrix of row sums of $\Omega$ such that the precisions for months with a larger number of neighbors will be higher than those with a relatively smaller number of neighbors. The priors allow for both a dependence across dimensions, $q \in (1,\ldots,Q)$, and months, $t \in (1,\ldots,T)$ in $\{\Theta, \Gamma_1, \ldots, \Gamma_L\}$.   The data estimate the marginal dependence among the $(y_{c1},\ldots,y_{cQ})$ from both shared links of some $(y_{cq})$ (for establishments, $i\{c\}~c \in (1,\ldots,n_{\mbox{\tiny{cases}}})$) to the industry indexed random effects and also from the by-dimension and month dependencies within each matrix-variate parameter. The form of the priors and the algorithm to sample from the pseudo posterior distribution of parameters $\{\Theta, \Gamma_1, \ldots, \Gamma_L, \tau_1, \ldots, \tau_Q\}$ are described in Section 2 of the Supporting Information.

\subsection{Setup and Comparison Metric}
\label{sec:setup}
We compared the performance of our GWPP with the full sample pseudo posterior distribution. The sampling model for the simulated and real data were based on the hierarchical model in \eqref{psieff}. The sampling algorithm described in Section 2 of the Supplementary Material was used to obtain samples from every posterior distribution after appropriately choosing the sampling weights $w_{ij}$ in \eqref{psieff}. All sampling algorithms ran for 15,000 iterations. We collected every fifth sample after discarding the first 10,000 samples as burn-ins. The convergence of every chain to its stationary distribution was confirmed using trace plots, with stopping set by the fixed width criteria of \citet{Fleg:Jone:batc:2010}.

We more formally compare the GWPP to the full sample pseudo posteriors by computing a normalized total variation distance \eqref{tv}, which takes values in $[0,1]$; the accuracy metric is closer to $1$ for smaller total variation distance and a higher quality approximation of the full sample pseudo posterior by GWPP:
\begin{equation}\label{tv}
\text{accuracy }\{\overline{\pi}^{\pi}\left(\theta\vert y\right)\} = 1-\frac{1}{2}\mathop{\int}_{\Theta}\left\lvert \overline{\pi}^{\pi}\left(\theta\vert y\right) - \pi^{\pi}\left(\theta\vert y\right)\right\rvert d\theta \in \left[0,1\right],
\end{equation}
where $\overline{\pi}^{\pi}\left(\theta\vert y\right)$ denotes the density of GWPP and $\pi^{\pi}\left(\theta\vert y\right)$ denotes the density of full sample pseudo posterior distribution. We compute the distance metric by using its numerical form based on Riemannian summation. We measure the similarity of the GWPP and the pseudo posterior estimated on the full sample, rather than to the (proper) posterior distribution estimated on the full population because \citet{2015arXiv150707050S} have already shown, both in theory and simulations, that the full sample pseudo posterior distribution contracts on the population posterior distribution.

\subsection{Simulated Data}
Consider the sampling model of the Current Employment Statistics survey data described in Section \ref{model}. We fixed $N, T$, $Q$, and $L$ defined in Section \ref{model} as $10,000$, $10$, $2$, and $1$, which excluded any industry-indexed random effects without loss of generality. We fixed $r$ at 0.9 to simulate $P_3$ using \eqref{car}. Given $t$, we fixed $\text{var}(y_{it1})$, $\text{var}(y_{it2})$, and $\text{cov}(y_{it1}, y_{it2})$ at 0$\cdot$5, 2, and 0$\cdot$6 to define $P_2$ ($i=1, \ldots, N$). We first simulated $\Theta$ using  \eqref{priorGamma} and then generated the population level response $q$ for establishment $i$ at time $t$, $y_{itq}$ ($i=1, \ldots, N$; $q=1, \ldots, Q$), as follows:
\begin{align}
  \label{eq:2}
  y_{itq} \mid \tau_q, \, \psi_{tq} &\ind \text{NB} \left\{ \tau_q,  \exp  \left( \psi_{tq} \right) \right\},\quad
  \psi_{tq} = 5 + \theta_{tq},
\end{align}
where $i=1, \ldots, 10,000,\; t=1, \ldots, 10,\; q=1, 2$, $\tau_1=5$, and $\tau_2=10$. The covariance matrices $P_2$ and $P_3$ induced dependence in $\theta_{tq}$s  across $t$s and the two $q$s.

We first generated a finite population according to \eqref{eq:2}, then subsequently drew two informative samples from the finite population of the $N$ establishments with the inclusion probability for each establishment $i$ set to be proportional to $y_{i \cdot \cdot}=  \sum_{t=1}^T \sum_{q=1}^Q y_{itq}$. The sampled data are composed of response values for both variables and all $10$ time points for each establishment included in each sample.  We sampled $n = fN$ of the $N$ establishments of the finite population in each of the two samples for sampling fraction, $f \in \{0.4, 0.6\}$. Establishments contained in each of the two samples were next randomly partitioned into $K$ subsets, each of equal size, $m = n/K$, where $K \in \{5, 10\}$.

We next obtained samples of parameters under \eqref{psieff} from the finite population posterior distribution, full sample pseudo posterior distribution, and our method in every replication. A new finite population and associated set of samples was generated in each simulation replication.  We set $w_i=1$ to obtain parameter draws from the finite population posterior distribution. Parameter draws from the full sample pseudo posterior distribution of size $n$ were estimated by setting $w_i = n \, y_{i \cdot \cdot}^{-1} (\sum_{i=1}^n y_{i \cdot \cdot}^{-1})^{-1}$ ($i=1, \ldots, n$), which normalizes the sampling weights to sum to $n$ for regulation of the uncertainties of estimated parameters.  We drew parameter samples from subset pseudo posterior $j$ by normalizing $w_{ij} = n \, y_{i \cdot \cdot}^{-1} (\sum_{i=1}^n y_{i \cdot \cdot}^{-1})^{-1}$ for every establishment $i$ in the $j$th subset in \eqref{psieff}, which regulates the amount of uncertainty in each subset $j$ to approximate that in the full sample.  Next, we
used the samples from the subset pseudo posterior distributions for each parameter to obtain a combined sample for the corresponding one-dimensional marginal.  We performed this step for each parameter of $\Theta$, and $\Gamma_1, \ldots, \Gamma_L$. For every such marginal, we combined the collection of samples from $K$ subset pseudo posterior distributions using the PIE algorithm\footnote{Software available at \url{https://github.com/david-dunson/divide-conquer-bayes}} \citep{2016arXiv160504029L}.  This simulation setup was replicated 10 times.

The GWPP showed excellent performance in approximating the full sample pseudo posterior distribution for both $K=5$ and $K=10$.  Figure~\ref{sim-dens} demonstrates that estimated pseudo posterior densities our method under both $K=5$ and $K=10$ very closely approximate the full sample pseudo posterior, both in locations and the amount of estimated uncertainties.  Table \ref{simtv} displays computed accuracies (of the GWPP compared to the full sample posterior) for each of the $\theta_{qt}$s, which are all close to $1$. Assumptions (A1)--(A8) were satisfied in our simulation example, so the results of our method were \emph{not} sensitive to the size of the subsets $K$, agreeing with Theorem \ref{main_wass}. The conditions of Theorem \ref{main_subset} were easier to satisfy when $K=5$ than when $K=10$ due to a larger subset size, resulting in higher accuracy for the GWPP with $K=5$ in some cases. In all our simulation examples, the GWPP required only 25\% of the memory resources used by the full sample pseudo posterior; sampling from the subset pseudo posterior and full sample pseudo posterior distributions respectively required 8GB and 32GB of memory resources. Estimation of the GWPP was about 10-times faster than the full sample pseudo posterior in run-time (Figure \ref{time}).  The relative improvement in computation time may be further enhanced in the case that the data analyst has a larger computational budget with more compute nodes.  We demonstrated robust performance as we increased the number of subsets from $K=5$ to $K=10$, based on our computational budget, though we would expect continued robust estimation performance with larger $K$, while memory usage and computation time would further improve, so long as we retain a reasonable subset sample size.

\begin{table}[H]
\def~{\hphantom{0}}
\caption{\it The accuracy \eqref{tv} of the GWPP for the marginals of $\Theta$ averaged across 10 simulation replications. The maximum Monte Carlo error over 10 simulation replications was 0.045.}
{\small
\begin{tabular}{rcccccccccccccccccccc}
  $(q, t)$ & (1, 1) & (2, 1) & (1, 2) & (2, 2) & (1, 3) & (2, 3) & (1, 4) & (2, 4) & (1, 5) & (2, 5) \\
  $K=5, f=60\%$  & 0$\cdot$96 & 0$\cdot$97 & 0$\cdot$94 & 0$\cdot$97 & 0$\cdot$95 & 0$\cdot$96 & 0$\cdot$95 & 0$\cdot$97 & 0$\cdot$96 & 0$\cdot$96 \\
  $K=10, f=60\%$  & 0$\cdot$95 & 0$\cdot$96 & 0$\cdot$94 & 0$\cdot$96 & 0$\cdot$94 & 0$\cdot$96 & 0$\cdot$95 & 0$\cdot$96 & 0$\cdot$94 & 0$\cdot$96 \\
  $K=5, f=40\%$  & 0$\cdot$96 & 0$\cdot$96 & 0$\cdot$93 & 0$\cdot$96 & 0$\cdot$93 & 0$\cdot$96 & 0$\cdot$95 & 0$\cdot$96 & 0$\cdot$94 & 0$\cdot$96 \\
  $K=10, f=40\%$  & 0$\cdot$95 & 0$\cdot$96 & 0$\cdot$94 & 0$\cdot$95 & 0$\cdot$93 & 0$\cdot$95 & 0$\cdot$95 & 0$\cdot$96 & 0$\cdot$94 & 0$\cdot$96 \\
  $(q, t)$ & (1, 6) & (2, 6) & (1, 7) & (2, 7) & (1, 8) & (2, 8) & (1, 9) & (2, 9) & (1, 10) & (2, 10) \\
  $K=5, f=60\%$  & 0$\cdot$95 & 0$\cdot$96 & 0$\cdot$96 & 0$\cdot$96 & 0$\cdot$95 & 0$\cdot$96 & 0$\cdot$95 & 0$\cdot$96 & 0$\cdot$96 & 0$\cdot$96 \\
  $K=10, f=60\%$ & 0$\cdot$94 & 0$\cdot$95 & 0$\cdot$95 & 0$\cdot$97 & 0$\cdot$95 & 0$\cdot$96 & 0$\cdot$96 & 0$\cdot$96 & 0$\cdot$96 & 0$\cdot$96 \\
  $K=5, f=40\%$ & 0$\cdot$95 & 0$\cdot$97 & 0$\cdot$94 & 0$\cdot$96 & 0$\cdot$94 & 0$\cdot$96 & 0$\cdot$93 & 0$\cdot$95 & 0$\cdot$95 & 0$\cdot$97 \\
  $K=10, f=40\%$ & 0$\cdot$96 & 0$\cdot$95 & 0$\cdot$93 & 0$\cdot$95 & 0$\cdot$93 & 0$\cdot$95 & 0$\cdot$92 & 0$\cdot$97 & 0$\cdot$95 & 0$\cdot$96 \\
\end{tabular}
}
\label{simtv}
\end{table}

\begin{figure}[H]
  \subfloat[$f=60\%$]{
    \includegraphics[width=0.9\textwidth]{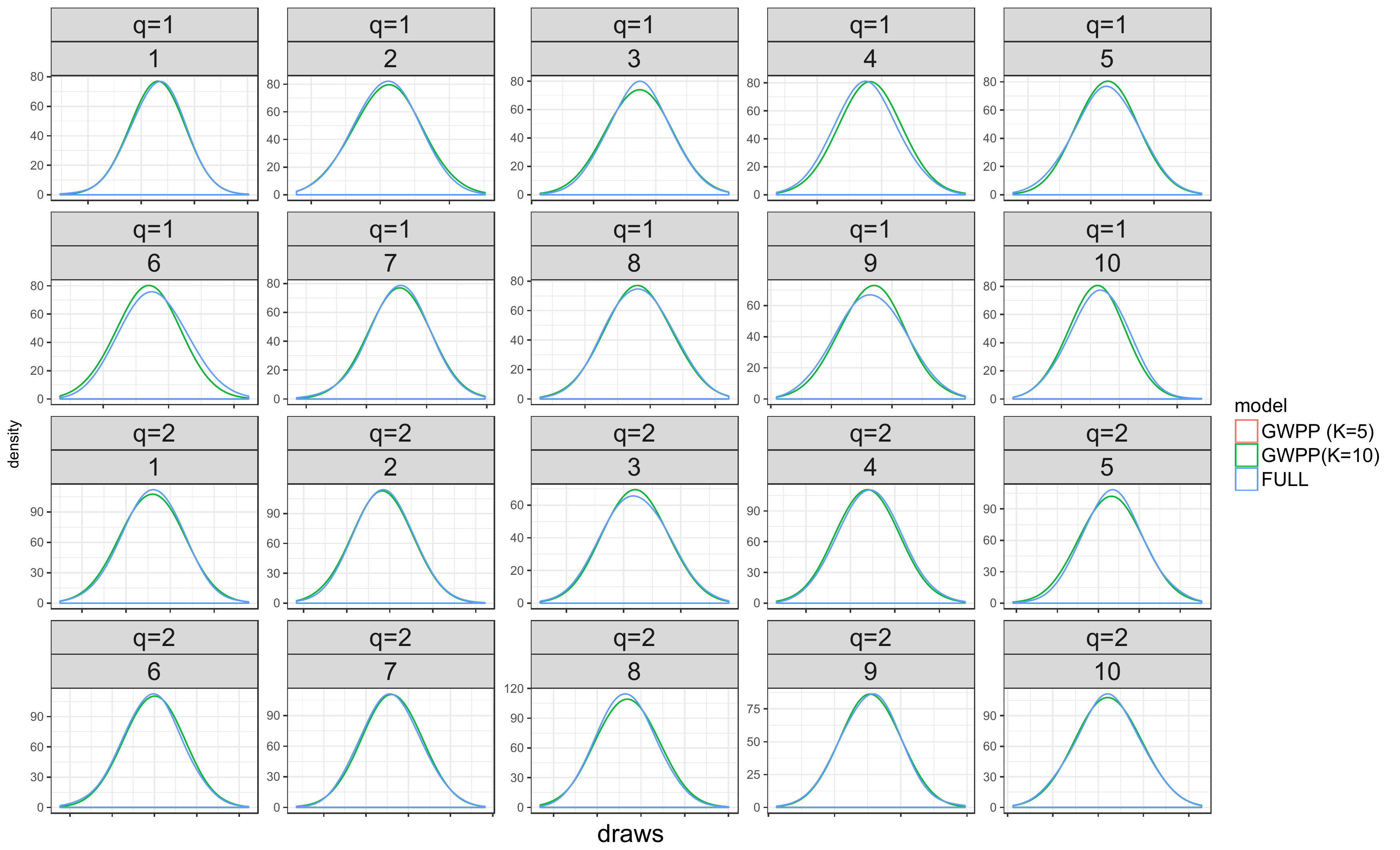}
    \label{fig:60}}\\
  \subfloat[$f=40\%$]{
  \includegraphics[width=0.9\textwidth]{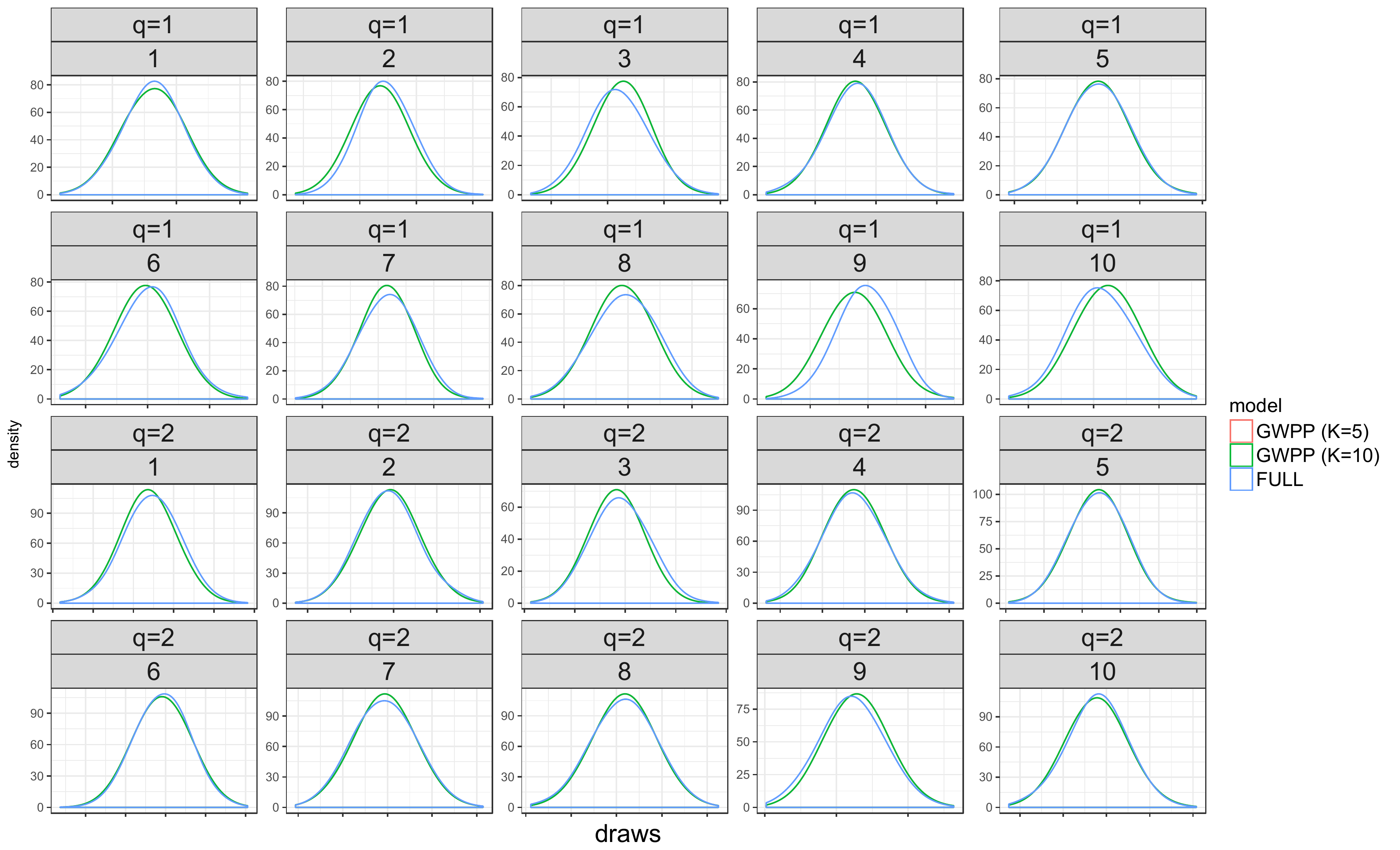}
    \label{fig:40}}
  \caption{Comparison of the full sample pseudo posterior density (FULL) and GWPP density in a simulation replication, where $f$ is the sampling fraction as in assumption (A8) and $K$ is the number of subsets.}
  \label{sim-dens}
\end{figure}
\FloatBarrier

\begin{figure}[H]
  \includegraphics[width=0.9\textwidth]{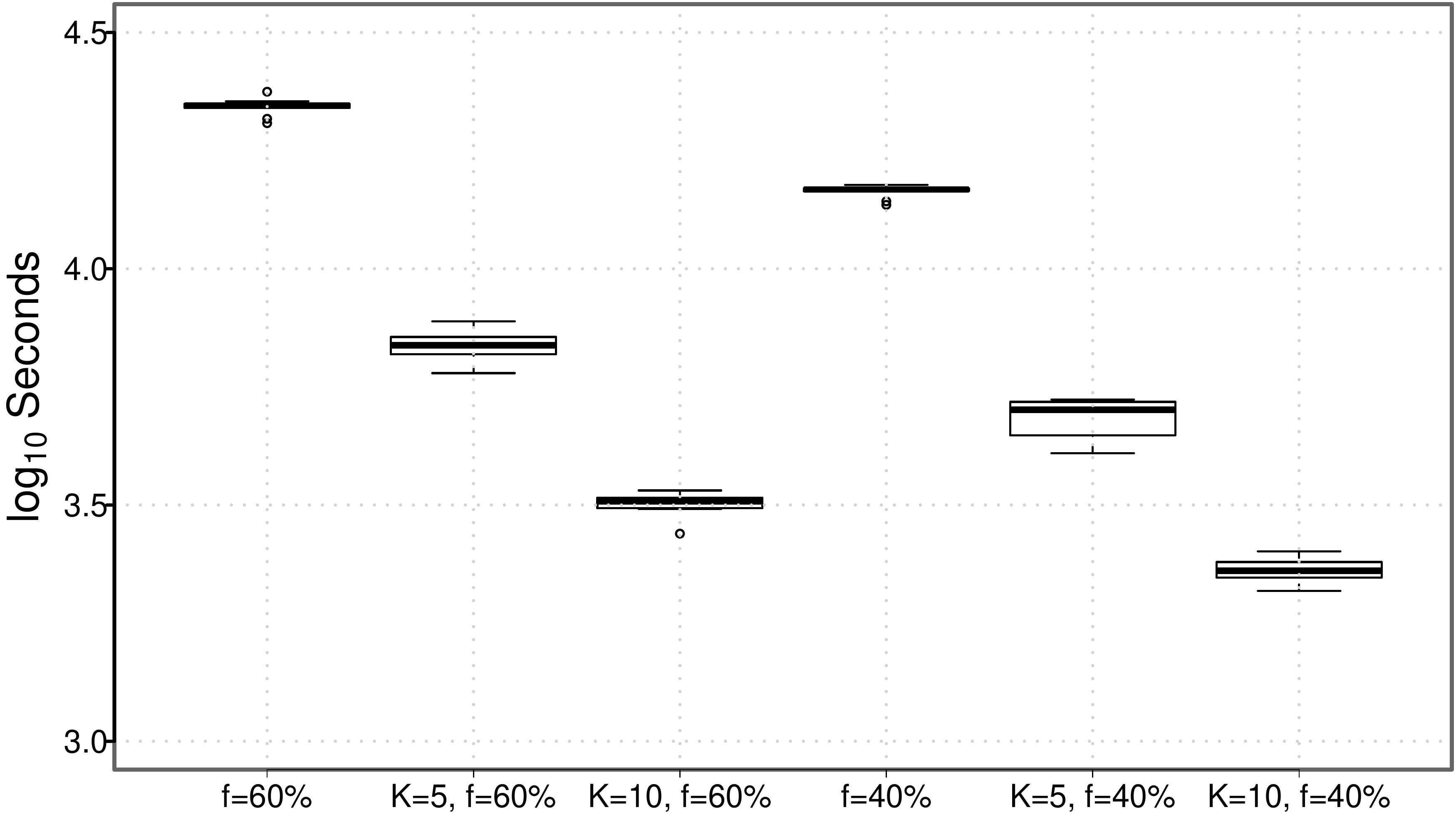}
  \caption{Computation time for the full sample pseudo posteriors and GWPPs across 10 simulation replications. The x-axis labels with $K$ correspond to GWPPs and those without $K$ correspond to full sample pseudo posteriors, $f$ is the sampling fraction as in assumption (A8), and $K$ is the number of subsets.}
  \label{time}
\end{figure}
\FloatBarrier

\subsection{Application to Current Employment Statistics Survey Data}

The survey data for California had $n=39360$ business establishments, each providing responses over multiple months for a total of $n_{\mbox{\tiny{cases}}} = 297000$ establishment-month cases.  We used the establishment-month case observations in the state of California for our comparisons because it was computationally feasible to estimate the full sample pseudo posterior distribution using the hierarchical model in \eqref{psieff}.  Our goal was to demonstrate that the GWPP could be used as an alternative for the full sample pseudo posterior distribution for inference on model parameters and for imputation of missing responses.

We randomly allocated the $n$ establishments to $K  = 4$ subsets of roughly equal numbers of establishments, $ {m} = (9017,9140,9082,9151)$ associated with $n_{\mbox{\tiny{cases}}} = (72841,74009,73702,74122)$ establishment-month case observations.  We divided the $n$ establishments into $L = 23$ industry-indexed strata and conducted simple random sampling within each stratum to populate the subsets.  Stratified selection ensured that all $L = 23$ industry super-sectors were linked to one or more establishments in each subset. We selected $K = 4$ subsets to accommodate our budget for computation and to ensure that $m_{j}$ was sufficiently large such that the conditions for our Theorem \ref{main_subset} were satisfied.

The GWPP provided a good approximation to the full sample pseudo posterior distribution. While the resulting GWPP and full sample pseudo posterior distributions were somewhat more complex than those in the simulation study, the two sets were, nevertheless, fairly similar in the masses of the distributions across various industry super-sectors (Figure~\ref{densities}).  The scaling of the subset pseudo posteriors under generalized stochastic approximation worked very well in that the spread of generalized Wasserstein pseudo posterior and full sample pseudo posterior distributions were similar, suggesting that uncertainty quantification using the two posterior distributions would be similar. The full sample distributions were, however, slightly more peaked than those of the GWPP.  This similarity among the distributional masses was further confirmed using the metric in \eqref{tv}, which showed that the generalized Wasserstein posterior was more than 81\% accurate in approximating the marginals of the full sample pseudo posterior for $\Theta$ (Table~\ref{tvtab}).

The GWPP also showed excellent performance in imputation, which combines the effects of the model parameters. Our model in \eqref{psieff} involved specification of a relatively large number of parameters to parameterize the log means, $\psi_{cq}$s. We constructed the means of our negative binomial model on the data scale (which is relevant for our purpose), $\exp\left(\psi_{cq}\right)$, using  \eqref{psieff}.  These means were used to impute the missing $y_{cq}$s from the posterior predictive distribution constructed from the GWPP for $\theta_{q j}$s and $\gamma_{\ell q j}$s.  The distribution of the posterior mean values of $\exp\left(\psi_{cq}\right)$s associated with the missing responses were nearly identical for the full sample pseudo posterior distribution and the GWPP (Figure~\ref{mis}).

\begin{figure}[H]
  \subfloat[$\Gamma$]{
    \includegraphics[width=0.8\textwidth]{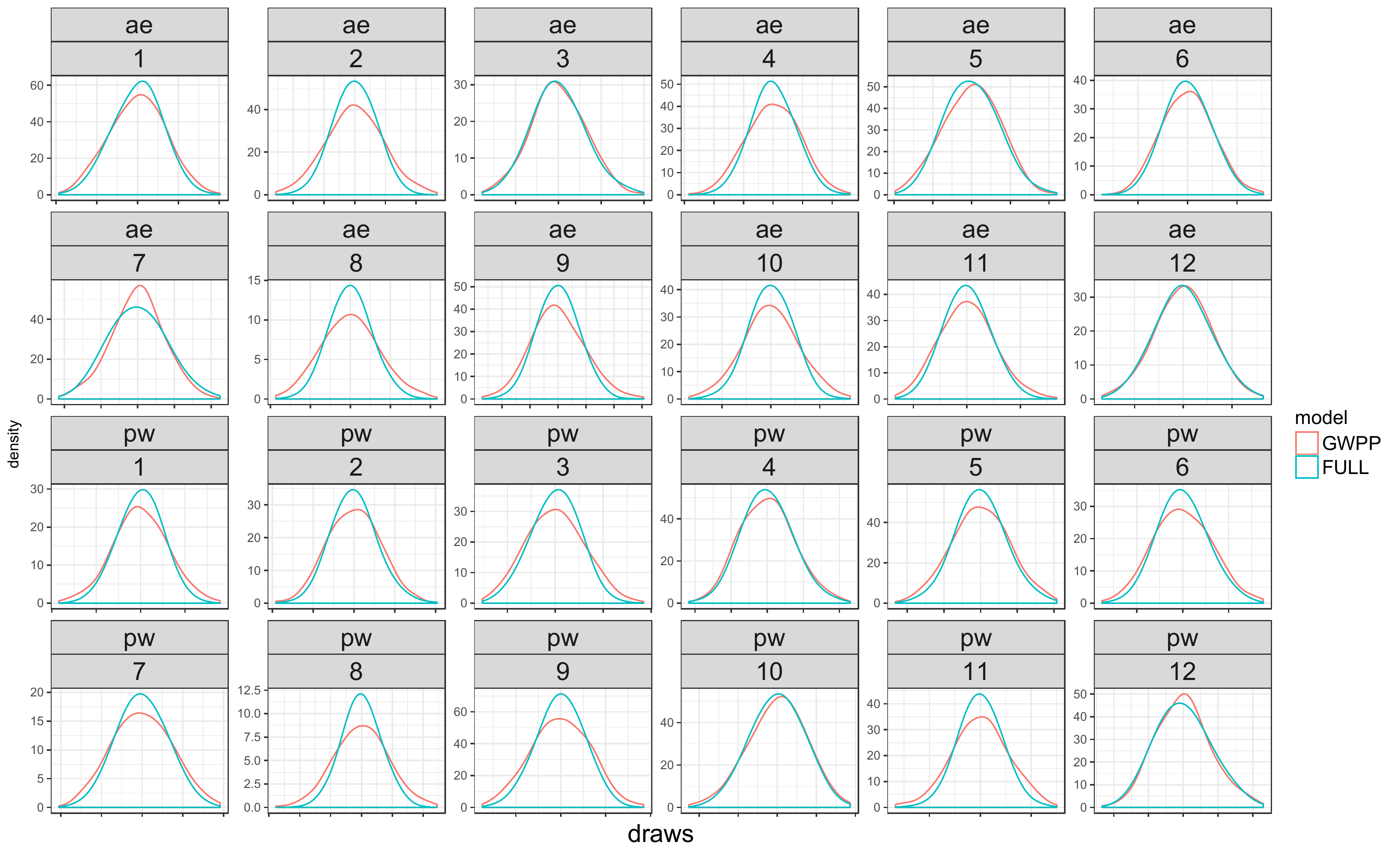}
    \label{fig:gam}}\\
  \subfloat[$\Theta$]{
  \includegraphics[width=0.8\textwidth]{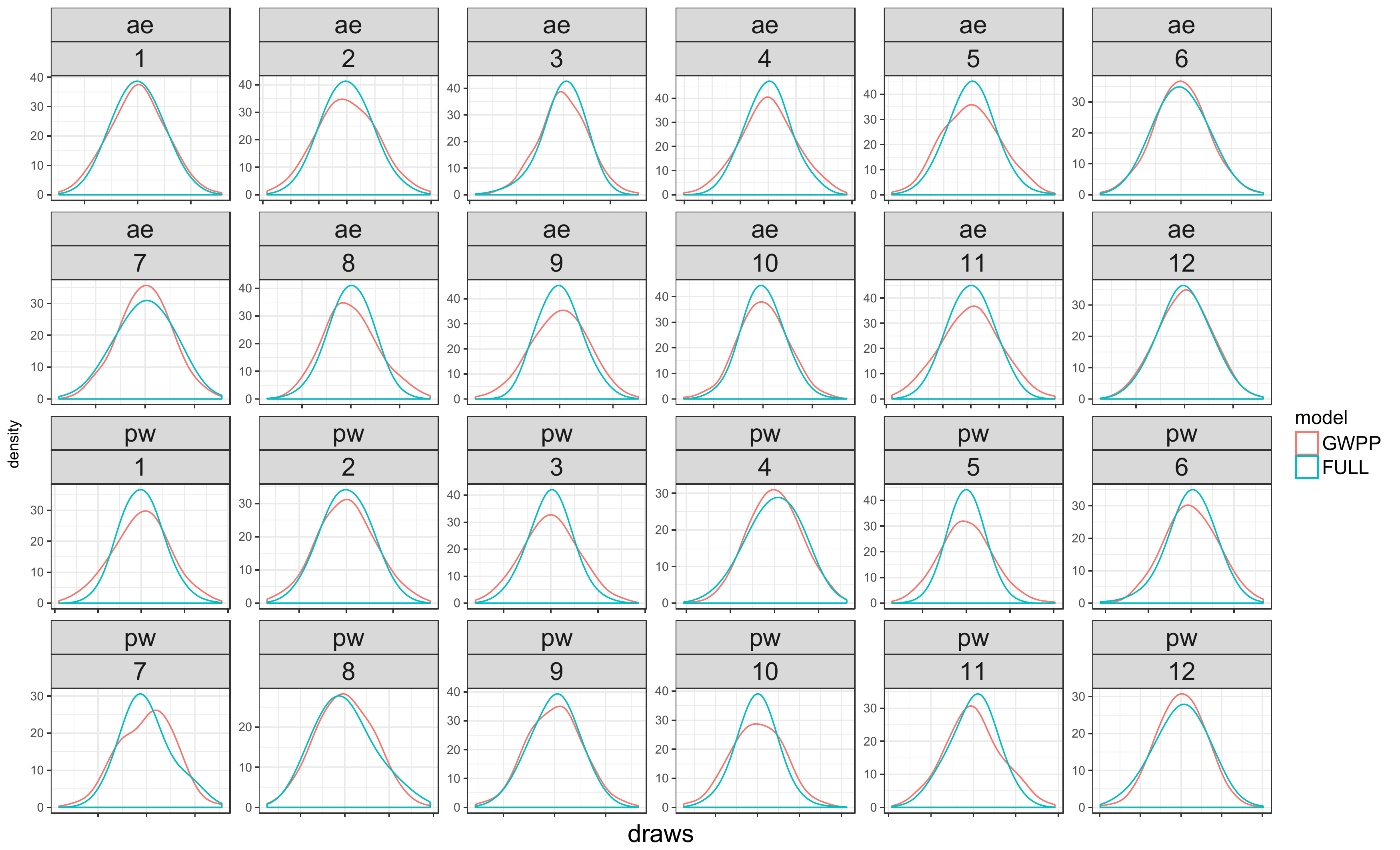}
    \label{fig:thet}}
  \caption{Comparison of the full sample pseudo posterior density (FULL) and GWPP density in application to the CES sample data.  Each plot panel compares the GWPP density under $K = 4$ (in red) to the full sample pseudo posterior density (in turquoise) for the selected parameters. Each panel represents a month ($1-12$) and a variable (ae,pw), where ae represents all employees and pw represents production workers. The top plot panels are for all $\{\gamma_{\ell q j}\}, \ell = $ the Professional \& Technical industry super-sector.  The bottom plot panels include the intercept parameters, $\{\theta_{q j}\}$.}
  \label{densities}
\end{figure}
\FloatBarrier

\begin{table}[H]
\def~{\hphantom{0}}
\caption{\it The accuracy \eqref{tv} of the GWPP for the marginals $\theta_{qt}$ ($q=1, 2$; $t=1, \ldots, 12$) in application to the CES sample. }
{\small
\begin{tabular}{cccccccccccc}
  $\theta_{11}$ & $\theta_{21}$ & $\theta_{12}$ & $\theta_{22}$ & $\theta_{13}$ & $\theta_{23}$ & $\theta_{14}$ & $\theta_{24}$ & $\theta_{15}$ & $\theta_{25}$ & $\theta_{16}$ & $\theta_{26}$\\
  0$\cdot$92 & 0$\cdot$84 & 0$\cdot$88 & 0$\cdot$90 & 0$\cdot$91 & 0$\cdot$86 & 0$\cdot$88 & 0$\cdot$94 & 0$\cdot$86 & 0$\cdot$81 & 0$\cdot$95 & 0$\cdot$89\\
  $\theta_{17}$ & $\theta_{27}$ & $\theta_{18}$ & $\theta_{28}$ & $\theta_{19}$ & $\theta_{29}$ & $\theta_{1\,10}$ & $\theta_{2\,10}$ & $\theta_{1\,11}$ & $\theta_{2\,11}$ & $\theta_{1\,12}$ & $\theta_{2\,12}$\\
  0$\cdot$94 & 0$\cdot$82 & 0$\cdot$86 & 0$\cdot$92 & 0$\cdot$84 & 0$\cdot$92 & 0$\cdot$90 & 0$\cdot$84 & 0$\cdot$86 & 0$\cdot$87 & 0$\cdot$95 & 0$\cdot$95 \\
\end{tabular}
}
\label{tvtab}
\end{table}

\begin{figure}[H]
\includegraphics[width=0.99\linewidth]{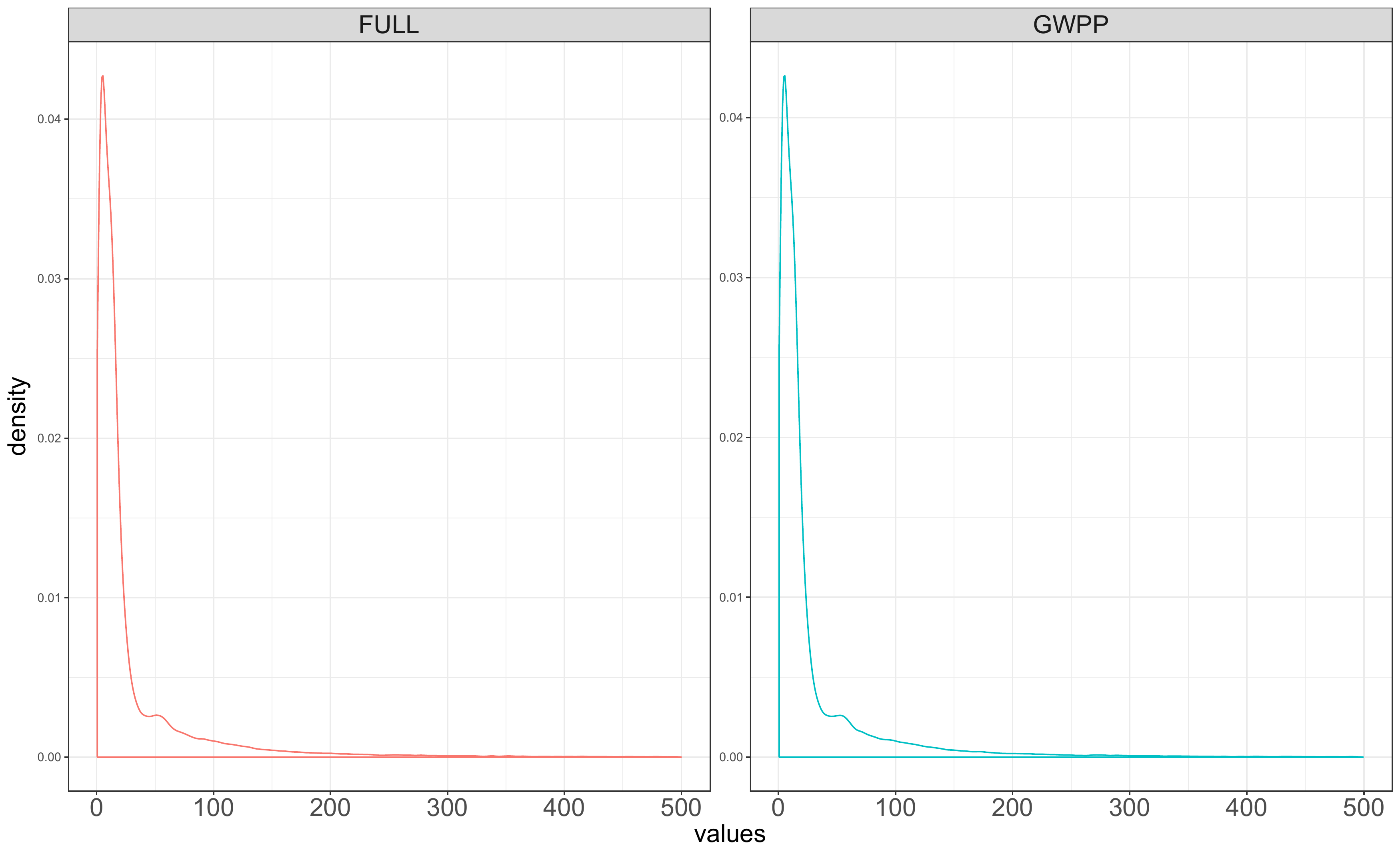}
\caption{Comparison of distribution of $n_{\tiny{\mbox{miss}}} \times 1$ posterior means, $\exp\left(\psi_{cq}\right)$, estimated for application to the CES sample using the GWPP and full sample pseudo posterior distributions (FULL). The right-hand density plot presents the distribution of the $n_{\tiny{\mbox{miss}}}$, $\{\exp\left(\psi_{cq}\right)\}$, linked to missing $\{y_{cq}\}$
estimated from the GWPP, while the left-hand density plot presents the distribution of the $n_{\tiny{\mbox{miss}}}$, $\{\exp\left(\psi_{cq}\right)\}$ estimated from the full sample pseudo posterior distribution. Approximately $25\%$ of $\{y_{cq}\}$ are missing.}
\label{mis}
\end{figure}
\FloatBarrier

\section{Concluding Remarks}\label{discussion}
We have extended stochastic approximation underlying WASP to dependent sample data collected under an informative sampling design.  We have demonstrated the contraction of both the subset pseudo posterior distributions using our sampling-weighted stochastic approximation and the computed GWPP under informative sampling where establishment marginal inclusion probabilities are correlated with the response.

The efficiency of the GWPP was critical in extending the inference on a low-dimensional parameter space to imputation on a parameter space of medium dimensions that provided sufficient flexibility for high quality imputation. Future areas of exploration include assessing feasibility of the GWPP under joint modeling of marginal sampling weights and the response of interest in a fully Bayesian construction, as contrasted with the plug-in pseudo posterior.

\begin{center}
\textbf{Supporting Information. Additional information for this article is available online}\label{suppA}
\end{center}
\begin{description}
  \item[Enabling Lemmas:] Enabling lemmas and proofs to support main theoretical results.
  \item[Model:] Hierarchical Model and Pseudo Posterior Formulations for Current Employment Statistics Survey Data.
  \begin{description}
  \item[Table:] Table S1 lists definitions for the $23$, $2-$ digit supersectors linking establishments in the CES
  \end{description}
\end{description}

  \bibliographystyle{agsm}
 \bibliography{mv_refs_sep2015_ss}

\appendix

\section{Proof of Theorem~\ref{main_subset}}\label{Appsubset}
\begin{proof}
We begin the proof in the same manner as in \citet{2015arXiv150805880S} by deconstructing the expectation of the squared Wasserstein distance from the pseudo posterior for subset $j$, $\Pi^{\pi}_{j}\left(\cdot  \mid y^{[j]}\delta_{\nu[j]}\right)$, to the delta measure at $\theta_{0}$, into two parts.  We recall assumption ~\nameref{compact} that $\Theta$ is compact, so that the sieve, $\Theta_{N_{\nu}}$, specified in \citet{ghosal2007} equals the entire space, $\Theta$, and we are able to bound, $\rho\left(\theta,\theta_{0}\right) < B_{0}$:
\begin{align}
  &\mathbb{E}_{P_{\theta_{0}},P_{\nu}} \left[ W_{2}^{2} \left\{ \Pi^{\pi}_{j}\left(\cdot  \mid y^{[j]}\delta_{\nu[j]}\right),\delta_{\theta_{0}}(\cdot) \right\} \right] = \mathbb{E}_{P_{\theta_{0}},P_{\nu}}\mathop{\int}_{\theta\in\Theta}\rho^{2}\left(\theta,\theta_{0}\right)\Pi_{j}^{\pi}\left(d\theta\mid y_{[j]}  {\delta}_{\nu [j]}\right)\nonumber\\
  &\leq \mathbb{E}_{P_{\theta_{0}},P_{\nu}}\mathop{\int}_{\left\{\theta:\rho\left(\theta,\theta_{0}\right) \leq c_{1}\epsilon_{M_{\nu}}\right\}}\rho^{2}\left(\theta,\theta_{0}\right)\Pi_{j}^{\pi}\left(d\theta\mid y_{[j]}  {\delta}_{\nu [j]}\right)+\mathbb{E}_{P_{\theta_{0}},P_{\nu}}\mathop{\int}_{\left\{\rho\left(\theta,\theta_{0}\right)> c_{1}\epsilon_{M_{\nu}}\right\}}\rho^{2}\left(\theta,\theta_{0}\right)\Pi_{j}^{\pi}\left(d\theta\mid y_{[j]}  {\delta}_{\nu [j]}\right)\nonumber\\
  &\leq \left(c_{1}\epsilon_{M_{\nu}}\right)^{2} + B_{0}^{2}\mathbb{E}_{P_{\theta_{0}},P_{\nu}}\Pi_{j}^{\pi}\left(\rho\left(\theta,\theta_{0}\right)>c_{1}\epsilon_{M_{\nu}}\mid y_{[j]}  {\delta}_{\nu [j]}\right)\label{postprob}.
\end{align}
We set constant, $\displaystyle c_{1} = \sqrt{\frac{2 r_{1} g_{2}\gamma^{2}}{q_{1} C_{L} }}$, and we note that it depends on the upper bound on the sampling weights, $\gamma$, specified in assumption ~\nameref{bounded} over all $(i,j) \in U_{\nu j}$ $(j=1, \ldots, k)$.  The additional constants $g_{2},r_{1}, q_{1}, C_{L}$ are specified in assumptions ~\nameref{distance} and ~\nameref{bounded} and in Lemmas~\ref{B5} and \ref{B6} in the online Supporting Information.

We next focus to bound the second term on the right-hand side of \eqref{postprob}.  The flow of the proof is most similar to Theorem 4.3 of \citet{2015arXiv150805880S} and Theorem 3 of \citet{2015arXiv150707050S}.  We extend these approaches to account for the taking of an informative random sample from the finite sub-populations, $U_{\nu j}$ $(j=1, \ldots, K)$. We first use assumption ~\nameref{distance} to bound the pseudo posterior with respect distance metric $\rho$ from above by the sampling-weighted, pseudo Hellinger distance,
\begin{align}
  &\Pi_{j}^{\pi}\left(\theta\in\Theta:\rho\left(\theta,\theta_{0}\right)>c_{1}\epsilon_{M_{\nu}}\mid y_{[j]}  {\delta}_{\nu [j]}\right)\\
  &\leq\Pi_{j}^{\pi}\left(\theta\in\Theta:h_{M_{\nu}}^{\pi}\left(\theta,\theta_{0}\right)>\sqrt{C_{L}}c_{1}\epsilon_{M_{\nu}}\mid y_{[j]}  {\delta}_{\nu [j]}\right),
\end{align}

We next bound the expectation with respect to the joint distribution, $\left(P_{\theta_{0}},P_{\nu}\right)$, of the pseudo posterior,
\begin{align}
  &\displaystyle\Pi_{j}^{\pi}\left(\theta\in\Theta:h_{M_{\nu}}^{\pi}\left(\theta,\theta_{0}\right)>\sqrt{C_{L}}c_{1}\epsilon_{M_{\nu}}\mid y_{[j]}  {\delta}_{\nu [j]}\right)=\nonumber\\
  &\displaystyle\frac{\displaystyle\mathop{\int}_{\left\{\theta\in\Theta: h_{M_{\nu}}^{\pi}\geq\sqrt{C_{L}}c_{1}\epsilon_{M_{\nu}}\right\}}\displaystyle\mathop{\prod}_{i=1}^{M_{\nu}}\frac{p^{\pi}_{\theta ji}}{p^{\pi}_{\theta_{0} ji}}\Pi\left(d\theta\right)}{\displaystyle\mathop{\int}_{\theta\in\Theta}\mathop{\displaystyle\prod}_{i=1}^{M_{\nu}}\frac{p^{\pi}_{\theta ji}}{p^{\pi}_{\theta_{0} ji}}\Pi\left(d\theta\right)}.\label{postmass}
\end{align}

We may bound the probability mass from below for some minimum value of the denominator of \eqref{postmass} using assumption ~\nameref{thickness} and Lemma~\ref{B6} in the online Supplemental Information such that with probability greater than or equal to $1-\left(r_{2}M_{\nu}\epsilon_{M_{\nu}}^{2}\right)^{-1}$,
\begin{equation}
  \displaystyle\mathop{\int}_{\theta\in\Theta}\mathop{\displaystyle\prod}_{i=1}^{M_{\nu}}\frac{p^{\pi}_{\theta ji}}{p^{\pi}_{\theta_{0} ji}}\Pi\left(d\theta\right) > \exp\left(-r_{1} N_{\nu}\epsilon_{M_{\nu}}^{2}\right).
\end{equation}

We next bound the numerator of \eqref{postmass}, from above, in $(P_{\theta_{0}},P_{\nu})$-probability, using assumptions ~\nameref{size}, ~\nameref{bounded}, and Lemma~\ref{B5} in the online Supplemental Information where the numerator,
\begin{align}
   \mathop{\int}_{\left\{\theta\in\Theta: h_{M_{\nu}}^{\pi}\geq\sqrt{C_{L}}c_{1}\epsilon_{M_{\nu}}\right\}}\displaystyle\mathop{\prod}_{i=1}^{M_{\nu}}\frac{p^{\pi}_{\theta ji}}{p^{\pi}_{\theta_{0} ji}}\Pi\left(d\theta\right) &\leq \exp\left(-\frac{q_{1}  M_{\nu} C_{L} c_{1}^{2}\epsilon_{M_{\nu}}^{2}}{2\gamma}\right)\label{plugintau}\\                                                                                                                                                                     &\leq \exp\left(-2 r_{1}g_{2} \gamma M_{\nu} \epsilon_{M_{\nu}}^{2}\right)\label{pluginc1}\\
   &\leq \exp\left(-2 r_{1} N_{\nu} \epsilon_{M_{\nu}}^{2}\right)\label{MtoN},
 \end{align}
 The inequality in \eqref{plugintau} results from plugging in $\tau = \sqrt{C_{L}}c_{1}\epsilon_{M_{\nu}}$ into the result for Lemma 1 in the online Supplemental Information.  The inequality in \eqref{pluginc1} results from plugging in for $c_{1}$.  We used $N_{\nu} \leq g_{2}\gamma M_{\nu}$ from assumption ~\nameref{bounded} to achieve the inequality in \eqref{MtoN}.

The lower bound of \eqref{MtoN} is realized with probability at least
 \begin{align}
   &1 - 4\exp\left(-\frac{q_{2}M_{\nu} C_{L}c_{1}^{2}\epsilon_{M_{\nu}}^{2}}{\gamma}\right)\label{plugtauprob}\\
   &= 1 - 4\exp\left(-\frac{r_{1}q_{2}g_{2}\gamma M_{\nu}\epsilon_{M_{\nu}}^{2}}{q_{1}}\right)\label{numerprob},
 \end{align}
where we, again, plug in for $\tau$ for the probability bound of Lemma 1 in the Supplementary Material to achieve \eqref{plugtauprob} and for $c_{1}$ to achieve \eqref{numerprob}. Then with probability at least $1- 4\exp\left(-\frac{r_{1}q_{2}g_{2}\gamma M_{\nu}\epsilon_{M_{\nu}}^{2}}{q_{1} }\right)-\left(r_{2}M_{\nu}\epsilon_{M_{\nu}}^{2}\right)^{-1}$
 \begin{align}
   &\Pi_{j}^{\pi}\left(\theta\in\Theta:\rho\left(\theta,\theta_{0}\right)>c_{1}\epsilon_{M_{\nu}}\mid y_{[j]}  {\delta}_{\nu [j]}\right)\nonumber\\
   &\leq \exp\left(-2 r_{1} N_{\nu} \epsilon_{M_{\nu}}^{2}\right)\times \exp\left(r_{1} N_{\nu}\epsilon_{M_{\nu}}^{2}\right)\nonumber\\
   &\leq \exp\left(-r_{1} N_{\nu} \epsilon_{M_{\nu}}^{2}\right)
 \end{align}

Let the event, $\displaystyle A^{\pi}_{M_{\nu}} = \left\{\Pi_{j}^{\pi}\left(\theta\in\Theta:\rho\left(\theta,\theta_{0}\right)>c_{1}\epsilon_{M_{\nu}}\mid y_{[j]}  {\delta}_{\nu [j]}\right) \leq \exp\left(-r_{1} N_{\nu} \epsilon_{M_{\nu}}^{2}\right)\right\}$, which we use
 to establish the $L_{1}$ bound,
 \begin{align}
   &\mathbb{E}_{P_{\theta_{0}},P_{\nu}}\left[\Pi_{j}^{\pi}\left(\theta\in\Theta:\rho\left(\theta,\theta_{0}\right)>c_{1}\epsilon_{M_{\nu}}\mid y_{[j]}  {\delta}_{\nu [j]}\right)\right]\nonumber\\
   &=\mathbb{E}_{P_{\theta_{0}},P_{\nu}}\biggl[\mathbf{I}\left(A^{\pi}_{M_{\nu}}\right)\times\Pi_{j}^{\pi}\left(\theta\in\Theta:\rho\left(\theta,\theta_{0}\right)>c_{1}\epsilon_{M_{\nu}}\mid y_{[j]}  {\delta}_{\nu [j]}\right)\nonumber\\
   &+ \mathbf{I}\left(\left[A^{\pi}_{M_{\nu}}\right]^{c}\right)\times\Pi_{j}^{\pi}\left(\theta\in\Theta:\rho\left(\theta,\theta_{0}\right)>c_{1}\epsilon_{M_{\nu}}\mid y_{[j]}  {\delta}_{\nu [j]}\right)\biggr]\nonumber\\
   &\leq \exp\left(-r_{1} N_{\nu} \epsilon_{M_{\nu}}^{2}\right) + P_{P_{\theta_{0}},P_{\nu}}\left(\left[A^{\pi}_{M_{\nu}}\right]^{c}\right)\nonumber\\
   &\leq \exp\left(-r_{1} N_{\nu} \epsilon_{M_{\nu}}^{2}\right) + 4\exp\left(-\frac{r_{1}q_{2}g_{2}\gamma M_{\nu}\epsilon_{M_{\nu}}^{2}}{q_{1}}\right)+\frac{1}{r_{2}M_{\nu}\epsilon_{M_{\nu}}^{2}}\nonumber\\
   &\leq \exp\left(-r_{1} N_{\nu} \epsilon_{M_{\nu}}^{2}\right) + 4\exp\left(-\frac{r_{1}q_{2} N_{\nu}\epsilon_{M_{\nu}}^{2}}{q_{1}}\right)+\frac{1}{r_{2}M_{\nu}\epsilon_{M_{\nu}}^{2}}\\
   &\leq 5\exp\left(-r_{1}c_{4} N_{\nu}\epsilon_{M_{\nu}}^{2}\right)+\frac{1}{\left(r_{2}M_{\nu}\epsilon_{M_{\nu}}^{2}\right)}\label{finalbound},
 \end{align}
 where $c_{4} = \min\left(\frac{q_{2}}{q_{1}},1\right)$. The first term in \eqref{finalbound} dominates because our loss of independence prevents the use of Bernstein's inequality as leveraged in \cite{massart2007concentration} and \citet{2015arXiv150805880S} to get an exponential lower bound on the denominator of \eqref{postmass}.  Returning to the decomposition of the $W_{2}$ distance in \eqref{postprob},
 \begin{align}
   &\mathbb{E}_{P_{\theta_{0}},P_{\nu}}W_{2}^{2}\left(\Pi^{\pi}_{j}\left(\cdot  \mid y^{[j]}\delta_{\nu[j]} \right),\delta_{\theta_{0}}(\cdot)\right)\nonumber\\
   &\leq \left(c_{1}\epsilon_{M_{\nu}}\right)^{2} + B_{0}\left[\frac{1}{r_{2}M_{\nu}\epsilon_{M_{\nu}}^{2}} + 5\exp\left(-r_{1}c_{4} N_{\nu}\epsilon_{M_{\nu}}^{2}\right)\right],
 \end{align}
 uniformly for all $j = 1,\ldots,K$, for constants, $r_{1}\geq \frac{\left(c_{\pi}g_{2} + 3\left(\kappa\gamma\right)^{-1}\right)}{g_{1}},~r_{2}= \frac{1}{\left[c_{3} + 1+\gamma\right]} \leq 1,~c_{1} = \sqrt{\frac{2 r_{1} g_{2}\gamma^{2}}{q_{1} C_{L} }}$ and $c_{4} = \min\left(\frac{q_{2}}{q_{1}},1\right)$.
 \end{proof}

\newpage
\section{Proof of Theorem~\ref{main_wass}}\label{Appwasp}
 \begin{proof}
   We bound the probability using Lemma B.7 from \citet{2015arXiv150805880S}, such that for any constant, $c_{5}$, which is a function of the sampling design constants $\gamma$ and $c_{3}$,

  \begin{align}
       &\mbox{Pr}_{\theta_{0},P_{\nu}}\left(W_{2}\left\{ \overline{\Pi}^{{\pi}}(\cdot \mid \{y_i : \delta_{\nu i} = 1, i = 1, \ldots, N_{\nu}\}),\delta_{\theta_{0}}\left(\cdot\right) \right\} > \sqrt{c_{5}}\epsilon_{M_{\nu}}\right)\\
       &\leq P_{\theta_{0},P_{\nu}}\left(\frac{1}{K}\mathop{\sum}_{j = 1}^{K}W_{2} \left\{ \Pi^{\pi}_{j}\left(\cdot  \mid y_{[j]}\delta_{\nu[j]}\right),\delta_{\theta_{0}}\left(\cdot\right) \right\} > \sqrt{c_{5}}\epsilon_{M_{\nu}}\right)\\
       &\overset{(i)}{\leq}\frac{1}{c_{5}\epsilon_{M_{\nu}}^{2}}E_{\theta_{0},P_{v}}\left[\left(\frac{1}{K}\mathop{\sum}_{j = 1}^{K}W_{2}\left(\Pi^{\pi}_{j}\left(\cdot  \mid y_{[j]}\delta_{\nu[j]} \right),\delta_{\theta_{0}}\left(\cdot\right)\right)\right)^{2}\right]\\
        &\overset{(ii)}{\leq}\frac{1}{c_{5}\epsilon_{M_{\nu}}^{2}} E_{\theta_{0},P_{v}} \left[\left( \frac{1}{K} \mathop{\sum}_{j = 1}^{K} W^2_{2}\left\{ \Pi^{\pi}_{j}\left(\cdot  \mid y_{[j]}\delta_{\nu[j]} \right),\delta_{\theta_{0}}\left(\cdot\right)\right\}  \right)\right] \\
       &\overset{(iii)}{\leq}\frac{1}{c_{5}\epsilon_{M_{\nu}}^{2}} \left( \frac{1}{K} \mathop{\sum}_{j = 1}^{K} E_{\theta_{0},P_{v}} \left[ W^2_{2}\left\{ \Pi^{\pi}_{j}\left(\cdot  \mid y_{[j]}\delta_{\nu[j]} \right),\delta_{\theta_{0}}\left(\cdot\right)\right\} \right] \right) \\
       &\overset{(iv)}{\leq}\frac{1}{c_{5}\epsilon_{M_{\nu}}^{2}} \left\{ c_{1}^{2}\epsilon_{M_{\nu}}^{2} + \frac{B_{0}}{r_{2}M_{\nu}\epsilon_{M_{\nu}}^{2}} + 5B_{0}\exp\left(-r_{1}c_{4} N_{\nu}\epsilon_{M_{\nu}}^{2}\right) \right\} \leq\frac{3B_{0}c_{1}\epsilon_{M_{\nu}}^{2}}{c_{5}\epsilon_{M_{\nu}}^{2}} = \frac{3B_{0}c_{1}}{c_{5}}
  \end{align}
  for $\nu$ sufficiently large, where $(i)$ follows from Markov's inequality, $(ii)$ follows from Jensen's inequality applied to simple averages, $(iii)$ follows from the linearity of expectation and $(iv)$ follows from Theorem \ref{main_subset}.

 \end{proof}
 
\section{Enabling Lemmas} \label{AppEnabling}
We now further generalize the two key lemmas constructed for $inid$ data from \citet{2015arXiv150805880S} to dependent data acquired under informative sampling: 1. Lemma $B.5$, a concatenation inequality; 2. Lemma $B.6$, which bounds the normalization constant of the pseudo posterior distribution.  Our generalized Lemmas ~\ref{B5} and \ref{B6}, together, enable the main two results on consistency of the subset pseudo posterior distributions and of the barycenter distribution composed from those subset posteriors. In all cases, size indices, $(m,n)$ in \citet{2015arXiv150805880S} are replaced with $(M_{\nu},N_{\nu})$ to refer to the finite population, and the associated, $(m_{\nu},n_{\nu})$, the observed sample taken from the sub-populations.

To prove Lemma~\ref{B5}, we extend \citet{2015arXiv150805880S} Lemmas $B.1-B.4$ to unequally weighted likelihood contributions by random weights (with respect to $P_{\nu}$) in the sequel.  Our new Lemma~\ref{B5} is the desired extension of the concentration inequality to dependent, informative sampling governed by $P_{\nu}$. Let $Z_{ji}(\theta) := \log \left( p_{\theta ji}/p_{\theta_{0} ji}\right)$, denote the logarithm of the likelihood ratio and $\tilde{Z}_{ji}(\theta) = \max(Z_{ji},-\tau)$ to denote the lower truncated version, as outlined in \citet{wong1995}, with $\tau>0$ to be selected in the sequel.   Construct a sampling weighted log-likelihood with, $Z_{ji}^{\pi}(\theta) = \delta_{\nu ji}/\pi_{\nu ji} \times Z_{ji}(\theta)$, its associated lower truncated version, $\tilde{Z}_{ji}^{\pi}(\theta) = \max(\delta_{\nu ji}/\pi_{\nu ji}\times Z_{ji},-\tau)$ and its vectorization, $\tilde{Z}_{j}^{\pi}(\theta) = \left(\tilde{Z}_{j1}^{\pi},\ldots,\tilde{Z}_{jM_{\nu}}^{\pi}\right)$.

\begin{lemma}\label{B1}
(Revised \citet{2015arXiv150805880S} Lemma B.1 under informative sampling).
Let $c_{1\tau} = 2e^{-\frac{\tau}{2}}/(1-e^{-\frac{\tau}{2}})^{2}$.  Then, for any $\theta\in\Theta$,
\begin{equation}
\frac{1}{M_{\nu}}\mathop{\sum}_{i=1}^{M_{\nu}}\mathbb{E}_{P_{\theta_{0}},P_{\nu}}~\tilde{Z}_{ji}^{\pi}(\theta) \leq -\left(1-c_{1\tau}\right)h^{2}_{M_{\nu}}\left(\theta,\theta_{0}\right)
\end{equation}
\end{lemma}
\begin{proof}
\begin{align*}
\mathbb{E}_{P_{\theta_{0}},P_{\nu}}\left(\tilde{Z}_{ji}^{\pi}\right) &= \mathbb{E}_{P_{\theta_{0}}}\left[\mathbb{E}_{P_{\nu}}\left(\max\left(\frac{\delta_{\nu ji}}{\pi_{\nu ji}}Z_{ji},-\tau\right)\bigg\vert \mathcal{A}_{\nu j}\right)\right]\\
&=
\mathbb{E}_{P_{\theta_{0}}}\left[\max\left(Z_{ji},-\tau\right)\right]\leq (1-c_{1\tau})h^{2}\left(p_{\theta_{1 ji}},p_{\theta_{2 ji}}\right),
\end{align*}
where $\mathcal{A}_{\nu j}$ is the sigma field of information in the subset population $j$ indexed by $\nu$. We use $\mathbb{E}_{P_{\nu}}(\delta_{\nu ji}\big\vert\mathcal{A}_{\nu j}) = \pi_{\nu ji}$, which cancels the denominator.
\end{proof}

\begin{lemma}\label{B2}
(Revised \citet{2015arXiv150805880S} Lemma B.2 under informative sampling).
Let $c_{2\tau} = \left(e^{-\frac{\tau}{2}}-1-\frac{\tau}{2}\right)/(1-e^{-\frac{\tau}{2}})^{2}$.  For any $t>0$, integer $\ell \geq 2$ and any $\theta\in\Theta$ restricted to $h_{M_{\nu}}\left(\theta,\theta_{0}\right)\leq r$,
\begin{equation}
\frac{1}{M_{\nu}}\mathop{\sum}_{i=1}^{M_{\nu}}\mathbb{E}_{P_{\theta_{0}},P_{\nu}}\bigg\lvert
\frac{\tilde{Z}_{ji}^{\pi}(\theta)}{2\sqrt{2c_{2\tau}\gamma}r}\bigg\rvert^{\ell} \leq \frac{\ell !}{2}
\left(\frac{1}{\sqrt{2c_{2\tau}\gamma}r}\right)^{\ell-2}
\end{equation}
\end{lemma}
\begin{proof}
\begin{equation}
\lvert\tilde{Z}_{ji}^{\pi}\rvert = \bigg\lvert\max\left(\frac{\delta_{\nu ji}}{\pi_{\nu ji}}Z_{ji},-\tau\right)\bigg\rvert =
\frac{1}{\pi_{\nu ji}}\bigg\lvert\max\left(\delta_{\nu ji}Z_{ji},-\tau\pi_{\nu ji}\right)\bigg\rvert \mathop{\leq}^{(i)} \frac{1}{\pi_{\nu ji}}\bigg\lvert\max\left(Z_{ji},-\tau\right)\bigg\rvert \mathop{\leq}^{(ii)} \gamma\lvert \tilde{Z}_{ji} \rvert,
\end{equation}
where $(i)$ results from $\lvert\max(y,-x_{2})\rvert \leq \lvert\max(y,-x_{1})\rvert$ for $x_{1}\geq x_{2}>0$ and $(ii)$ applies assumption~\nameref{bounded}.  We next apply Lemma 5 of \citet{wong1995} to,
\begin{equation}
\mathbb{E}_{P_{\theta_{0}},P_{\nu}}\left[\exp\left(\frac{\lvert\tilde{Z}_{ji}^{\pi}\rvert}{2\gamma}-1-
\frac{\lvert\tilde{Z}_{ji}^{\pi}\rvert}{2\gamma}\right)\right] \mathop{\leq}^{(iii)} \mathbb{E}_{P_{\theta_{0}}}\left[\exp\left(\frac{\lvert\tilde{Z}_{ji}\rvert}{2}-1-
\frac{\lvert\tilde{Z}_{ji}\rvert}{2}\right)\right]  \leq c_{2\tau}h^{2}\left(p_{\theta ji},p_{\theta_{0} ji}\right),
\end{equation}
where $(iii)$ results because $\exp(t/2)-1-t/2$ is increasing for $t\geq 0$.   We next apply the identity,
$\mathbb{E}\left[\exp\left(\frac{\lvert\tilde{Z}_{ji}^{\pi}\rvert}{2\gamma}-1-
\frac{\lvert\tilde{Z}_{ji}^{\pi}\rvert}{2\gamma}\right)\right]  \geq \frac{\mathbb{E}\left(\lvert\tilde{Z}_{ji}^{\pi}\rvert^{\ell}\right)}{\ell ! \gamma^{\ell}}$, which gives us,
\begin{equation}
\frac{1}{M_{\nu}}\mathop{\sum}_{i=1}^{M_{\nu}}\mathbb{E}_{P_{\theta_{0}},P_{\nu}}\bigg\lvert
\tilde{Z}_{ji}^{\pi}(\theta)\bigg\rvert^{\ell} \leq 2^{\ell}\ell ! c_{2\tau}\gamma^{\ell} r.
\end{equation}
Rearranging terms produces the result.
\end{proof}

\begin{lemma}\label{B3}
(Revised \citet{2015arXiv150805880S} Lemma B.3 under informative sampling).
Suppose assumption~\nameref{bounded}.  Let $\mathcal{P}_{j}\left(\Theta\right) =\left\{p_{\theta j1},\ldots,p_{\theta jM_{\nu}},~\theta\in\Theta\right\}$ and
\newline
$\tilde{\mathcal{Z}}_{j}^{\pi}\left(\Theta\right) = \left\{\tilde{Z}_{j1}^{\pi}\left(\theta\right),\ldots,\tilde{Z}_{jM_{\nu}}^{\pi}\left(\theta\right),~\theta\in\Theta\right\}$. For any $u>0$,
\begin{equation}
H_{[]}\left(u,\tilde{\mathcal{Z}}_{j}^{\pi},\norm{\cdot}\right) \leq H_{[]}\left(\frac{u}{2 \sqrt{\gamma} e^{\frac{\tau}{2}}},\mathcal{P}_{j},h_{M_{\nu}}\right)
\end{equation}
\end{lemma}
\begin{proof}
\begin{align}
&\mathbb{E}_{P_{\theta_{0}},P_{\nu}}\left[\tilde{Z}_{ji}^{\pi}\left(\theta_{1}\right)
-\tilde{Z}_{ji}^{\pi}\left(\theta_{2}\right)\right]^{2}\nonumber\\
&= \mathbb{E}_{P_{\theta_{0}}}\left[\mathbb{E}_{P_{\nu}}\left\{\max\left(\frac{\delta_{\nu ji}}{\pi_{\nu ji}}Z_{ji}(\theta_{1}),-\tau\right) - \max\left(\frac{\delta_{\nu ji}}{\pi_{\nu ji}}Z_{ji}(\theta_{2}),-\tau\right)\right\}^{2}\bigg\vert \mathcal{A}_{\nu j}\right]. \label{successive}
\end{align}
We will address $2$ cases for the value of $\left(\mathcal{Z}_{ji}^{\pi}(\theta_{1}),\mathcal{Z}_{ji}^{\pi}(\theta_{2})\right)$ to evaluate the integral of
\eqref{successive}.

Let $\max\left(\frac{\delta_{\nu ji}}{\pi_{\nu ji}}Z_{ji}(\theta_{1}),-\tau\right) = \frac{\delta_{\nu ji}}{\pi_{\nu ji}}Z_{ji}(\theta_{1})$ and $\max\left(\frac{\delta_{\nu ji}}{\pi_{\nu ji}}Z_{ji}(\theta_{2}),-\tau\right) = \frac{\delta_{\nu ji}}{\pi_{\nu ji}}Z_{ji}(\theta_{2})$.  Then the joint expectation in \eqref{successive} is equal to,
\begin{align}
&\mathbb{E}_{P_{\theta_{0}}}\left[\mathbb{E}_{P_{\nu}}\left\{\frac{\delta_{\nu ji}^{2}}{\pi_{\nu ji}^{2}}Z_{ji}(\theta_{1})^{2}-2\frac{\delta_{\nu ji}^{2}}{\pi_{\nu ji}^{2}}Z_{ji}(\theta_{1})Z_{ji}(\theta_{2}) + \frac{\delta_{\nu ji}^{2}}{\pi_{\nu ji}^{2}}Z_{ji}(\theta_{2})^{2}\bigg\vert\mathcal{A}_{\nu j}\right\}\right]\nonumber\\
&= \mathbb{E}_{P_{\theta_{0}}}\left[\frac{1}{\pi_{\nu ji}}\left({Z}_{ji}(\theta_{1}) - {Z}_{ji}(\theta_{2})\right)^{2}\right]\nonumber\\
&\leq \gamma\mathbb{E}_{P_{\theta_{0}}}\left({Z}_{ji}(\theta_{1}) - {Z}_{ji}(\theta_{2})\right)^{2},\label{gamresult}
\end{align}
where $\mathop{\mathbb{E}}_{P_{\nu}}\left(\delta_{\nu ji}^{2}\big\vert \mathcal{A}_{\nu j}\right) = \mathop{\mathbb{E}}_{P_{\nu}}\left(\delta_{\nu ji}\big\vert \mathcal{A}_{\nu j}\right) = \pi_{\nu ji}$.  The bound in \eqref{gamresult} results from assumption~\nameref{bounded}.

Next, we let $\max\left(\frac{\delta_{\nu ji}}{\pi_{\nu ji}}Z_{ji}(\theta_{1}),-\tau\right) = -\tau$, the lower-truncated level, and $\max\left(\frac{\delta_{\nu ji}}{\pi_{\nu ji}}Z_{ji}(\theta_{2}),-\tau\right) = \frac{\delta_{\nu ji}}{\pi_{\nu ji}}Z_{ji}(\theta_{2})$, as above.
Then the joint expectation in \eqref{successive} is equal to,
\begin{align}
&\mathbb{E}_{P_{\theta_{0}}}\left[\mathbb{E}_{P_{\nu}}\left\{\tau^{2} + 2\frac{\delta_{\nu ji}}{\pi_{\nu ji}}Z_{ji}(\theta_{2})\tau + \frac{\delta_{\nu ji}^{2}}{\pi_{\nu ji}^{2}}Z_{ji}(\theta_{2})^{2}\bigg\vert\mathcal{A}_{\nu j}\right\}\right]\nonumber\\
&= \mathbb{E}_{P_{\theta_{0}}}\left[\tau^{2} + 2Z_{ji}(\theta_{2})\tau + \frac{1}{\pi_{\nu ji}}Z_{ji}(\theta_{2})^{2}\right]\nonumber\\
&\leq \mathbb{E}_{P_{\theta_{0}}}\left[\tau^{2} + 2Z_{ji}(\theta_{2})\tau + \gamma Z_{ji}(\theta_{2})^{2}\right]\label{termresult}\\
&\leq \mathbb{E}_{P_{\theta_{0}}}\left[\gamma\tau^{2} + 2\gamma Z_{ji}(\theta_{2})\tau + \gamma Z_{ji}(\theta_{2})^{2}\right]\label{uptermresult}\\
&\leq \gamma\mathbb{E}_{P_{\theta_{0}}}\left({Z}_{ji}(\theta_{1}) - {Z}_{ji}(\theta_{2})\right)^{2}, \label{gam2result}
\end{align}
where we achieve \eqref{uptermresult} by noting that $\mathbb{E}_{P_{\theta_{0}}}\left[\gamma Z_{ji}(\theta_{2})^{2}\right] > 0$ in \eqref{termresult}, so that  $\mathbb{E}_{P_{\theta_{0}}}\left[\tau^{2} + 2Z_{ji}(\theta_{2})\tau\right] > 0$ since \eqref{termresult} is greater than $0$, which produces the inequality since
$\gamma \geq 1$.  We achieve \eqref{gam2result} by applying Lemma B.3 of \citet{wong1995} (after factoring out the $\gamma$).  By symmetry, switching the values for the maxima of the two expressions will produce the same result.   Finally, we note that if both sampling weighted random variables are truncated at $-\tau$, then the joint expectation is exactly equal to $0$ and may, therefore, by bounded by \eqref{gam2result}.

We next apply Lemma 3 of \citet{wong1995},
\begin{equation*}
\mathbb{E}_{P_{\theta_{0}},P_{\nu}}\left[\tilde{Z}_{ji}^{\pi}\left(\theta_{1}\right)
-\tilde{Z}_{ji}^{\pi}\left(\theta_{2}\right)\right]^{2} \leq \gamma\mathbb{E}_{P_{\theta_{0}}}\left({Z}_{ji}(\theta_{1}) - {Z}_{ji}(\theta_{2})\right)^{2} \leq 4\gamma e^{\tau}h^{2}\left(p_{\theta_{1} ji},p_{\theta_{2} ji}\right).
\end{equation*}

Averaging over
the subset $j$ population units, $i = 1,\ldots,M_{\nu}$ gives:
\begin{equation}
\left[\frac{1}{M_{\nu}}\mathop{\sum}_{i=1}^{M_{\nu}}\mathbb{E}_{P_{\theta_{0}},P_{\nu}}\left(
\tilde{Z}_{ji}^{\pi}(\theta_{1})-\tilde{Z}_{ji}^{\pi}(\theta_{2})\right)^{2} \right]^{\frac{1}{2}} \leq
\norm{\tilde{Z}_{ji}^{\pi}(\theta_{1})-\tilde{Z}_{ji}^{\pi}(\theta_{2})} \leq 2 \sqrt{\gamma} e^{\frac{\tau}{2}}h_{M_{\nu}}\left(\theta_{1},\theta_{2}\right),
\end{equation}
which implies the inequality result between bracketing entropies.
\end{proof}

\begin{lemma}\label{vandergeer}
(Special case of \citet{GeeLed13} Theorem 8) Let $j\in \{1,\ldots,K\}$ be fixed. Suppose assumption~\nameref{bounded} holds in the construction of a class of functions,
\newline
$\displaystyle\mathcal{F}^{\pi}_{j}=\left\{\mbf{f}^{\pi} =(\frac{\delta_{\nu j1}}{\pi_{\nu j1}} f_1(y_{j1}) ,\ldots,\frac{\delta_{\nu jM_{\nu}}}{\pi_{\nu jM_{\nu}}} f_{M_{\nu}}(y_{jM_{\nu}}))^T, \mathbf{y}_{j}=(y_{j1},\ldots,y_{jM_{\nu}})\in \otimes_{i=1}^{M_{\nu}} \mathcal{Y}_{ji}\right.$
\newline
$\bigg.\bm{\delta}_{\nu j} = \left(\delta_{\nu j1},\ldots,\delta_{\nu j M_{\nu}}\right)\in\{0,1\}^{M_{\nu}},~\bm{\pi}_{\nu j} = \left(\pi_{\nu j1},\ldots,\pi_{\nu jM_{\nu}}\right)\in (0,1]^{M_{\nu}} \bigg\}$ that satisfies\\
\noindent (i) $\mathop{\sup}_{\mbf{f}\in \mathcal{F}^{\pi}_{j} } \|\mbf{f}^{\pi}\| \leq 1$;\\
\noindent (ii) For any integer $\ell\geq 2$, $\mathop{\sup}_{\mbf{f}^{\pi}\in\mathcal{F}^{\pi}_{j}} |\mbf{f}^{\pi}|_{\ell}^{\ell} \leq \ell! C^{\ell-2}/2$, for some constant $C>0$;\\
Then for any $t>0$,
\begin{align*}
  \mbox{Pr}_{P_{\theta_{0}},P_{\nu}} \left(\sup_{\mbf{f}\in\mathcal{F}^{\pi}_{j}}\frac{1}{\sqrt{M_{\nu}}}\sum_{i=1}^{M_{\nu}} \left[\frac{\delta_{ji}}{\pi_{ji}} f_{i}(Y_{ji}) - E_{P_{\theta_{0}},P_{\nu}} \left\{ \frac{\delta_{ji}}{\pi_{ji}} f_{i}(Y_{ji}) \right\}\right] \geq \min_{S\in \mathbb{N}}R^{\pi}_S + \frac{36C(1+t)}{\sqrt{M_{\nu}}} + 24\sqrt{6t}\right) \leq 2e^{-t},
\end{align*}
where
\begin{align*}
R^{\pi}_S \equiv 2^{-S}\sqrt{M_{\nu}} + 14\sqrt{6} \sum_{s=0}^S 2^{-s} \sqrt{H_{[]}\left(2^{-s},\mathcal{F}^{\pi}_j,\|\cdot\|\right)} + \frac{36 C H_{[]}\left(1,\mathcal{F}^{\pi}_j,\|\cdot\|\right)}{\sqrt{M_{\nu}}}.
\end{align*}
We will refer to (i) and (ii) as the \textbf{Bernstein conditions}.
\end{lemma}
\begin{proof}
Theorem 8 in \citet{GeeLed13} construct a space of functions,$\mathcal{F}$, governed by a distribution, $\mathbb{P}$, without specifying a model.  Our result is, therefore, constructed as a special case of Theorem 8 by setting $\mathbb{P} = \Pr_{P_{\theta_0}, P_{\nu}}$, which is the distribution that governs our space of functions, $\mathcal{F}^{\pi}_j$.  Since assumption~\nameref{bounded} requires $1/\pi_{\nu ji} \leq \gamma$, we may construct an $\mbf{f}^{\pi}$ to meet the Bernstein conditions.
\end{proof}

\begin{lemma}\label{B5}
(Revised \citet{2015arXiv150805880S} Lemma B.5 under informative sampling)
Suppose assumptions \nameref{bounded} and \nameref{size}  hold.  Then for any $\tau > 0$, there exist positive constants $q_{1}, q_{2}$, that depend on $D_{1}, D_{2}$, such that for all subsets, $Y_{[j]}$, $j=1,\ldots,K$,
\begin{equation}\label{numerresult}
\mbox{Pr}_{\theta_{0},P_{\nu}}\left(\mathop{\sup}_{h_{M_{\nu}}^{\pi}\left(\theta,\theta_{0}\right)\geq \tau}
\mathop{\prod}_{i=1}^{M_{\nu}}\frac{p_{\theta ji}^{\pi}}{p_{\theta_{0} ji}^{\pi}}\geq\exp\left(-\frac{q_{1}M_{\nu}\tau^{2}}{\gamma}\right)\right)\\
\leq 4\exp\left(-\frac{q_{2}M_{\nu}\tau^{2}}{\gamma}\right),
\end{equation}
for $\nu$ sufficiently large.
\end{lemma}
\begin{proof}
The proof steps are identical to \citet{2015arXiv150805880S} Lemma B.5 and begin by constructing $R_{S}^{\pi}$ from Lemma~\ref{vandergeer} (the lower bound of the event) by defining a normalized, lower truncated log-likelihood ratio that satisfies conditions $(i)$ and $(ii)$ of that lemma where we replace $\tilde{Z}_{ji}$ used in \citet{2015arXiv150805880S} Lemma B.5 with $\tilde{Z}_{ji}^{\pi}$ (defined earlier) in the normalized class of functions,
\begin{equation*}
\hat{\mathcal{Z}}_{j}^{\pi}(r) = \left\{\frac{\tilde{Z}_{j}^{\pi}(\theta)}{2\sqrt{2c_{2\tau}\gamma}r}:\theta\in\Theta \text{ is restricted to } h_{M_{\nu}}^{\pi}\left(\theta,\theta_{0}\right)\leq r\right\},
\end{equation*}
where our revised Lemma~\ref{B2} adds a $\gamma$ in the denominator to accomplish the normalization and shows that conditions $(i)$ and $(ii)$ are satisfied with $C = 1/(\sqrt{2c_{2\tau}\gamma}r)$.

Lemma~\ref{B3} is next employed to bound the bracketing entropy terms of $R_{S}^{\pi}$ on the normalized, truncated space, $\hat{\mathcal{Z}}_{j}^{\pi}(r)$ in terms of the unnormalized, untruncated and non-sampling weighted space, $\mathcal{Z}_{j}(r) = \left(Z_{j1},\ldots,Z_{jM_{\nu}},~\theta\in\Theta \text{ is restricted to } h_{M_{\nu}}\left(\theta,\theta_{0}\right)\leq r\right)$, which updates \citet{2015arXiv150805880S} Lemma B.5 equations 30 and 31.  The integrand term for the two bracketing entropy computations updates from $\sqrt{2c_{2\tau}e^{-\tau}}r$ to $\sqrt{2c_{2\tau}\gamma^{2} e^{-\tau}}r$, where the first $\gamma$ results from undoing the normalization step (using Lemma~\ref{B2}) and the secong $\gamma$ derives from the upper bound on the bracketing Hellinger entropy of Lemma~\ref{B3}.   This result, together with our revised Lemma~\ref{B1}, implies that with probability at least $1-2 e^{-c_{3\tau}m_{\nu} r^{2}}$,
\begin{align*}
&\mathop{\sup}_{\theta\in\Theta: r\leq h_{M_{\nu}}^{\pi}\left(\theta,\theta_{0}\right)\leq 2r}\frac{1}{M_{\nu}}\mathop{\sum}_{i=1}^{M_{\nu}}\tilde{Z}_{ji}^{\pi}(\theta) \\
&\leq\mathop{\sup}_{\theta\in\Theta: \frac{r}{\sqrt{\gamma}}\leq h_{M_{\nu}}\left(\theta,\theta_{0}\right)\leq 2\frac{r}{\sqrt{\gamma}}}\frac{1}{M_{\nu}}\mathop{\sum}_{i=1}^{M_{\nu}}\tilde{Z}_{ji}^{\pi}(\theta) \\
&\leq -\left(1-c_{1\tau}-8\sqrt{2c_{2\tau}\gamma^{2}}c_{4\tau\gamma}\right)\frac{r^{2}}{\gamma} + \frac{72}{M_{\nu}},
\end{align*}
where the smaller range of distance between $\theta$ and $\theta_{0}$ in the first inequality increases the sum and we have replaced $r$ by $r/\sqrt{\gamma}$ in the last inequality.  The constant, $c_{4\tau\gamma}$, in the last inequality updates $c_{4\tau}$ by replacing $c_{2\tau} e^{-\tau}$ with $c_{2\tau} \gamma^{2} e^{-\tau}$.

The desired result is achieved if $\left(1-c_{1\tau}-8\sqrt{2c_{2\tau}\gamma^{2}}c_{4\tau\gamma}\right) > 0$,
which we proceed to demonstrate by updating selected constants from \citet{2015arXiv150805880S} Lemma B.5.  Since $c_{1\tau}$ is decreasing in $\tau$, its value in \citet{2015arXiv150805880S} Lemma B.5 is maintained with their choice of $\tau$. The
$\sqrt{2}/2^{10}$ term in $c_{4\tau}$ of \citet{2015arXiv150805880S} is updated to $\sqrt{2}/(2^{10}\gamma^{2})$ by choosing a larger $S$ through, $2^{-(S+2)} \leq \sqrt{2c_{2\tau}\gamma^{2} e^{-\tau}}r/(2^{12}\gamma^{2})$.  Assumption~\nameref{size} gives
us $D_{2}\leq D_{1}^{2}/(2^{12}\gamma^{2})$, and we update $c_{3\tau}$ to $c_{3\tau\gamma} = 1/(2^{30}\gamma)$ (where this constant may be freely chosen) such that,
\begin{equation*}
8\sqrt{2c_{2\tau}\gamma^{2}}c_{4\tau\gamma} \leq 8\sqrt{60\gamma^{2}}\left\{\left[\frac{\sqrt{2}}{2^{10}\gamma^{2}} +\frac{56\sqrt{3}}{2^{12}\gamma^{2}} + \frac{144\sqrt{2}}{16\cdot 2^{24}\gamma^{2}}\right]\sqrt{\frac{30\gamma}{2^{10}}} + \frac{36}{\sqrt{58}\cdot 2^{30}\gamma^{2}} + 24\sqrt{\frac{6}{2^{30}\gamma^{2}}}\right\} < 0.377,
\end{equation*}
since $\gamma \geq 1$.  The rest of the proof is identical to \citet{2015arXiv150805880S} Lemma B.5 after replacing
$Z_{ji}(\theta)$ with $\frac{\delta{\nu ji}}{\pi_{\nu ji}}\times Z_{ji}(\theta)$ inside the event statement and replacing $r$ with $r/\sqrt{\gamma}$.

\end{proof}

\begin{lemma}\label{B6}
  Suppose assumptions (A1), (A4), (A7), and (A8) hold. Then there exist positive constants $r_{1},r_{2}= \frac{1}{\left[c_{3} + 1 +\gamma\right]} \leq 1$ that depend on $g_{1},g_{2},\kappa,c_{\pi},c_{3},\gamma$ such that for every subset, $Y_{[j]}$ ($j=1,\ldots,K$), and for any $t \geq \epsilon^{2}_{M_{\nu}}$,
\begin{equation}\label{denomresult}
\mbox{Pr}_{\theta_{0},P_{\nu}}\left\{\mathop{\int}_{\Theta}\displaystyle\mathop{\prod}_{i=1}^{N_{\nu}}\frac{p^{\pi}_{\theta ji}}{p^{\pi}_{\theta_{0} ji}}\Pi\left(d\Theta\right)\leq \exp\left[-r_{1}N_{\nu}t\right]\right\}
\leq \frac{1}{r_{2}M_{\nu}t},
\end{equation}
for $M_{\nu}$ sufficiently large, where the above probability is taken with the respect to the population generating distribution, $P_{\theta_{0}}$, and the sampling design distribution, $P_{\nu}$, jointly.
\end{lemma}

\begin{proof}
The proof generally follows the flow of \citet{2015arXiv150805880S} to simplify the probability statement on the left-hand side; only, we are not able to apply the Bernstein inequality (see Corollary 2.10 in \citet{massart2007concentration}) to formulate the bound for the resulting event probability because the $\left\{p^{\pi}_{\theta ji}/p^{\pi}_{\theta_{0} ji}\right\}_{i=1,\ldots,M_{\nu}}$ are \emph{not} independent due to the dependence induced among the $\bm{\delta}_{\nu j} = \left(\delta_{\nu 1},\ldots,\delta_{\nu M_{\nu}}\right)$ by the informative sampling distribution.   So we will follow the  strategy of \citet{Ghosal00convergencerates} (also used in \citet{zbMATH06017974}) and instead employ Chebyshev, along with a bound on the pairwise inclusion probabilities, to separate integration terms involving the sampling design distribution, $P_{\nu}$, from those involving the finite population generating distribution, $P_{\theta_{0}}$.

The constant, $r_{2}$, will depend on bounds, $\gamma$, and $c_{3}$, that express the efficiency of the sampling design.  The rate of convergence of the probability will be slower for sampling designs that produce samples with relatively larger information differences from the underlying finite population.

We first expand the event, $\Theta_{\epsilon_{M_{\nu}}}$ defined in \citet{2015arXiv150805880S} to incorporate informative sampling with respect to $P_{\nu}$,
\begin{equation*}
\Theta_{\epsilon_{M_{\nu}}}^{\pi} = \left\{\theta \in \Theta:\frac{1}{M_{\nu}}\mathop{\sum}_{i=1}^{M_{\nu}}\mathbb{E}_{P_{\theta_{0}},P_\nu}\exp\left(\log_{+}\frac{p^{\pi}_{\theta_{0} ji}}{p^{\pi}_{\theta ji}}\right) - 1 \leq \epsilon^{2}_{M_{\nu}}\right\},
\end{equation*}
which specifies an upper bound on the distance of $\theta$ from $\theta_{0}$.  The prior for this event may be bounded from below,
\begin{align}
\Pi\left\{\Theta_{\epsilon_{M_{\nu}}}^{\pi}\right\} &= \Pi^{\bm{\pi}}\left\{\theta \in \Theta: \frac{1}{M_{\nu}}\mathop{\sum}_{i=1}^{M_{\nu}}\mathbb{E}_{P_{\theta_{0}},P_{\nu}}\exp\left(\frac{\delta_{\nu ji}}{\pi_{\nu ji}}\log_{+}\frac{p_{\theta_{0} ji}}{p_{\theta ji}}\right) - 1 \leq \epsilon^{2}_{M_{\nu}}\right\}\\
&\geq \Pi\left\{\theta \in \Theta: \frac{1}{M_{\nu}}\mathop{\sum}_{i=1}^{M_{\nu}}\mathbb{E}_{P_{\theta_{0}}}\exp\left(\gamma\log_{+}\frac{p_{\theta_{0} ji}}{p_{\theta ji}}\right) - 1 \leq \epsilon^{2}_{M_{\nu}}\right\}\\
&\geq \exp\left(-c_\pi \gamma N_{\nu}\epsilon^{2}_{M_{\nu}}\right),
\end{align}

where the last expression results from using assumption (A4).  The result provides a lower bound on the prior mass assigned to the region defined by $\Theta_{\epsilon_{M_{\nu}}}^{\pi}$.  For $A \subseteq \Theta$, let $\Pi_{\epsilon_{M_{\nu}}^{\pi}}\left(A\right) = \Pi\left(A \cap \Theta_{\epsilon_{M_{\nu}}}^{\pi}\right)/\Pi\left(\Theta_{\epsilon_{M_{\nu}}}^{\pi}\right)$ be the prior measure that restricts $\Pi$ to the region of support, $\Theta_{\epsilon_{M_{\nu}}}^{\pi}$.

By Jensen's inequality,
\begin{align*}
\log\mathop{\int}_{\Theta}\mathop{\prod}_{i=1}^{M_{\nu}}\frac{p^{\pi}_{\theta ji}}{p^{\pi}_{\theta_{0} ji}}\Pi\left(d\theta\right) &\geq \mathop{\sum}_{i=1}^{M_{\nu}}\displaystyle\mathop{\int}_{\Theta}\log\frac{p^{\pi}_{\theta ji}}{p^{\pi}_{\theta_{0} ji}}\Pi\left(d\theta\right)\\
&= M_{\nu}\cdot\mathbb{P}_{M_{\nu}}\mathop{\int}_{\Theta}\frac{p^{\pi}_{\theta j}}{p^{\pi}_{\theta_{0} j}}\Pi\left(d\theta\right),
\end{align*}
where we recall that the last equation denotes the empirical expectation functional taken with respect to the joint distribution over population generating and informative sampling.  By Fubini,
\begin{align*}
\mathbb{P}_{M_{\nu}}\mathop{\int}_{\Theta}\log\frac{p^{\pi}_{\theta j}}{p^{\pi}_{\theta_{0} j}}\Pi\left(d\theta\right)
&= \mathop{\int}_{\Theta}\left[\mathbb{P}_{M_{\nu}}\log\frac{p^{\pi}_{\theta j}}{p^{\pi}_{\theta_{0} j}}\right]\Pi\left(d\theta\right)\\
&= \mathop{\int}_{\Theta}\left[\mathbb{P}_{M_{\nu}}\frac{\delta_{\nu j}}{\pi_{\nu j}}
\log\frac{p_{\theta j}}{p_{\theta_{0} j}}\right]\Pi\left(d\theta\right)\\
&= \mathop{\int}_{\Theta}\left[\mathbb{P}^{\pi}_{M_{\nu}}
\log\frac{p_{\theta j}}{p_{\theta_{0} j}}\right]\Pi\left(d\theta\right)\\
&= \mathbb{P}^{\pi}_{M_{\nu}}\mathop{\int}_{\Theta}\log\frac{p_{\theta j}}{p_{\theta_{0} j}}\Pi\left(d\theta\right),
\end{align*}
where we, again, apply Fubini.

Then, the probability statement in the result of Equation~\ref{denomresult} is bounded (from above) by,
\begin{align*}
&\mbox{Pr}_{\theta_{0},P_{\nu}}\left\{M_{\nu}\cdot\mathbb{P}^{\pi}_{M_{\nu}}\mathop{\int}_{\Theta}\log\frac{p_{\theta j}}{p_{\theta_{0} j}}
\Pi\left(d\theta\right) \leq -r_{1}N_{\nu}t\right\}\\
&\mathop{\leq}^{(i)}\mbox{Pr}_{\theta_{0},P_{\nu}}\left\{M_{\nu}\cdot\mathbb{P}^{\pi}_{M_{\nu}}\mathop{\int}_{\Theta_{\epsilon_{M_{\nu}}}^{\pi}}\log\frac{p_{\theta j}}{p_{\theta_{0} j}}
\Pi\left(d\theta\right) \leq -r_{1}N_{\nu}t\right\}\\
&\mathop{\leq}^{(ii)}\mbox{Pr}_{\theta_{0},P_{\nu}}\left\{M_{\nu}\cdot\Pi\left\{\Theta_{\epsilon_{M_{\nu}}}^{\pi}\right\}\cdot\mathbb{P}^{\pi}_{M_{\nu}}\mathop{\int}_{\Theta_{\epsilon_{M_{\nu}}}^{\pi}}\log\frac{p_{\theta j}}{p_{\theta_{0} j}}
\Pi_{\epsilon_{M_{\nu}}}\left(d\theta\right) \leq -r_{1}N_{\nu}t\right\}\\
&\mathop{\leq}^{(ii)}\mbox{Pr}_{\theta_{0},P_{\nu}}\left\{M_{\nu}\cdot\mathbb{P}^{\pi}_{M_{\nu}}\mathop{\int}_{\Theta_{\epsilon_{M_{\nu}}}^{\pi}}\log\frac{p_{\theta j}}{p_{\theta_{0} j}}
\Pi_{\epsilon_{M_{\nu}}^{\pi}}\left(d\theta\right) \leq -r_{1}N_{\nu}t + c_\pi \gamma N_{\nu}\epsilon^{2}_{M_{\nu}}\right\}\\
&\mathop{\leq}^{(iii)}\mbox{Pr}_{\theta_{0},P_{\nu}}\left\{M_{\nu}\cdot\mathbb{P}^{\pi}_{M_{\nu}}\mathop{\int}_{\Theta_{\epsilon_{M_{\nu}}}^{\pi}}\log\frac{p_{\theta_{0} j}}{p_{\theta j}}
\Pi_{\epsilon_{M_{\nu}}^{\pi}}\left(d\theta\right) \geq r_{1}g_{1}M_{\nu}\gamma t - c_\pi g_{2}M_{\nu}\gamma^{2}\epsilon^{2}_{M_{\nu}}\right\}\\
&\leq\mbox{Pr}_{\theta_{0},P_{\nu}}\left\{\mathbb{G}^{\pi}_{M_{\nu}}\mathop{\int}_{\Theta_{\epsilon_{M_{\nu}}}^{\pi}}\log\frac{p_{\theta_{0} j}}{p_{\theta j}}
\Pi_{\epsilon_{M_{\nu}}^{\pi}}\left(d\theta\right) \geq r_{1}g_{1}\sqrt{M_{\nu}}\gamma t - c_\pi g_{2}\sqrt{M_{\nu}}\gamma^{2}\epsilon^{2}_{M_{\nu}}\right.\\
&\left.-\sqrt{M_{\nu}}\mathop{\int}_{\Theta_{\epsilon_{M_{\nu}}}^{\pi}}\mathbb{E}_{P_{\theta_{0}},P_{\nu}}\log\frac{p_{\theta_{0} j}}{p_{\theta j}}\Pi_{\epsilon_{M_{\nu}}^{\pi}}\left(d\theta\right)\right\}\\
&\mathop{\leq}^{(iv)}\mbox{Pr}_{\theta_{0},P_{\nu}}\left\{\mathbb{G}^{\pi}_{M_{\nu}}\mathop{\int}_{\Theta_{\epsilon_{M_{\nu}}}^{\pi}}\log\frac{p_{\theta_{0} j}}{p_{\theta j}}
\Pi_{\epsilon_{M_{\nu}}^{\pi}}\left(d\theta\right) \geq r_{1}g_{1}\sqrt{M_{\nu}}\gamma t - \left(c_\pi g_{2}\gamma
-\left(\kappa\gamma\right)^{-1}\right)\sqrt{M_{\nu}}\gamma\epsilon^{2}_{M_{\nu}}\right\}
\end{align*}
where $(i)$ follows by making the integration region smaller; $(ii)$ from assumption (A4) on the reduced-size region $\Theta_{\epsilon_{M_{\nu}}}$; $(iii)$ results from application of assumption (A7) that globally bounds the vector of sampling inclusion probabilities away from $0$ for all $j = 1,\ldots,K$.  The integration on the right-hand side of $(iv)$ reduces, as follows:
\begin{eqnarray*}
\mathbb{E}_{P_{\theta_{0}},P_{\nu}}\mathop{\int}_{\Theta_{\epsilon_{M_{\nu}}}^{\pi}}\log\frac{p_{\theta_{0} j}}{p_{\theta j}}
\Pi_{\epsilon_{M_{\nu}}^{\pi}}\left(d\theta\right) = \mathop{\int}_{\Theta_{\epsilon_{M_{\nu}}}}\mathbb{E}_{P_{\theta_{0}},P_{\nu}}\log\frac{p_{\theta j}}{p_{\theta_{0} j}}
\Pi_{\epsilon_{M_{\nu}}^{\pi}}\left(d\theta\right)\\
\leq \kappa^{-1}\mathop{\int}_{\Theta_{\epsilon_{M_{\nu}}}^{\pi}}\mathbb{E}_{P_{\theta_{0}},P_{\nu}}\left[\exp\left(\kappa\log\frac{p_{\theta_{0} j}}{p_{\theta j}}
\right)-1\right]\Pi_{\epsilon_{M_{\nu}}^{\pi}}\left(d\theta\right) \leq \kappa^{-1}\epsilon^{2}_{M_{\nu}},
\end{eqnarray*}
where the first equality on the first line applies Fubini.  The first inequality on the second line applies the inequality, $x \leq \left(e^{\kappa x} - 1\right)/\kappa$ for $x \geq 0$ and the second inequality applies (A4) (for a single observation).

We now apply Chebyshev and Jensen's inequality to bound the probability,
\begin{subequations}\label{chebyshev}
\begin{align}
&\mbox{Pr}_{\theta_{0},P_{\nu}}\left\{\mathbb{G}^{\pi}_{M_{\nu}}\mathop{\int}_{\Theta_{\epsilon_{M_{\nu}}}^{\pi}}\log\frac{p_{\theta_{0} j}}{p_{\theta j}}\Pi_{\epsilon_{M_{\nu}}^{\pi}}\left(d\theta\right)\geq r_{1}g_{1}\sqrt{M_{\nu}}\gamma t - \left(c_\pi g_{2}\gamma
-\left(\kappa\gamma\right)^{-1}\right)\sqrt{M_{\nu}}\gamma\epsilon^{2}_{M_{\nu}}\right\}\nonumber\\
&\leq\frac{\Var\left[\mathop{\int}_{\Theta_{\epsilon_{M_{\nu}}}^{\pi}}\mathbb{G}^{\pi}_{M_{\nu}}\log\frac{p_{\theta_{0} j}}{p_{\theta j}}\Pi_{\epsilon_{M_{\nu}}^{\pi}}\left(d\theta\right)\right]}
{\left(r_{1}g_{1}\sqrt{M_{\nu}}\gamma t - \left(c_\pi g_{2}\gamma
-\left(\kappa\gamma\right)^{-1}\right)\sqrt{M_{\nu}}\gamma\epsilon^{2}_{M_{\nu}}\right)^{2}}\\
&\leq\frac{\mathbb{E}_{P_{\theta_{0}},P_{\nu}}\left[\mathop{\int}_{\Theta_{\epsilon_{M_{\nu}}}^{\pi}}\mathbb{G}^{\pi}_{M_{\nu}}\log\frac{p_{\theta_{0} j}}{p_{\theta j}}\Pi_{\epsilon_{M_{\nu}}^{\pi}}\left(d\theta\right)\right]^{2}}
{\left(r_{1}g_{1}\sqrt{M_{\nu}}\gamma t - \left(c_\pi g_{2}\gamma
-\left(\kappa\gamma\right)^{-1}\right)\sqrt{M_{\nu}}\gamma\epsilon^{2}_{M_{\nu}}\right)^{2}}\label{chebyshev:var}\\
&\leq\frac{\displaystyle\mathop{\int}_{\Theta_{\epsilon_{M_{\nu}}}^{\pi}}\left[\mathbb{E}_{P_{\theta_{0}},P_{\nu}}
\left(\mathbb{G}^{\pi}_{M_{\nu}}\log\frac{p_{\theta_{0} j}}{p_{\theta j}}\right)^{2}\right]\Pi_{\epsilon_{M_{\nu}}^{\pi}}\left(d\theta\right)}
{\left(r_{1}g_{1}\sqrt{M_{\nu}}\gamma t - \left(c_\pi g_{2}\gamma
-\left(\kappa\gamma\right)^{-1}\right)\sqrt{M_{\nu}}\gamma\epsilon^{2}_{M_{\nu}}\right)^{2}}\label{chebyshev:e2},
\end{align}
\end{subequations}
where $\mathbb{E}_{P_{\theta_{0}},P_{\nu}}\left(\cdot\right)$ denotes the expectation with respect to the joint distribution over population generation and sampling (from that population) without replacement.   We apply Jensen's inequality in Equation~\ref{chebyshev:var} and use $\mathbb{E}\left(X^{2}\right) > \Var\left(X\right)$ in the third inequality, stated in Equation~\ref{chebyshev:e2}.  We now bound the expectation inside the square brackets on the right-hand side of Equation~\ref{chebyshev:e2}, which is taken with respect to this joint distribution.  In the sequel, define $\mathcal{A}_{\nu j} = \sigma\left(Y_{j1},\ldots,Y_{jM_{\nu}}\right)$ as the sigma field of information potentially available for the $M_{\nu}$ units in population, $U_{\nu j}$.

\begin{subequations}
\begin{align}
&\mathbb{E}_{P_{\theta_{0}},P_{\nu}}\left[\mathbb{G}^{\pi}_{M_{\nu}}\log\frac{p_{\theta_{0} j}}{p_{\theta j} }\right]^{2}\\
&= \mathbb{E}_{P_{\theta_{0}},P_{\nu}}\left[\sqrt{M_{\nu}}\left(\mathbb{P}^{\pi}_{M_{\nu}} - \mathbb{P}_{M_{\nu}}\right)\log\frac{p_{\theta_{0} j}}{p_{\theta j}}-\sqrt{M_{\nu}}\left(\mathbb{P}_{0} - \mathbb{P}_{M_{\nu}}\right)\log\frac{p_{\theta_{0} j}}{p_{\theta j}}\right]^{2}\\
&= \mathbb{E}_{P_{\theta_{0}},P_{\nu}}\left[\sqrt{M_{\nu}}\left(\mathbb{P}^{\pi}_{M_{\nu}} - \mathbb{P}_{M_{\nu}}\right)\log\frac{p_{\theta_{0} j}}{p_{\theta j}}- \mathbb{G}_{M_{\nu}}\log\frac{p_{\theta_{0} j}}{p_{\theta j}}\right]^{2}\\
&\leq M_{\nu}\mathbb{E}_{P_{\theta_{0}},P_{\nu}}\left[\left(\mathbb{P}^{\pi}_{M_{\nu}} - \mathbb{P}_{M_{\nu}}\right)\log\frac{p_{\theta_{0} j}}{p_{\theta j}}\right]^{2} + \mathbb{E}_{P_{\theta_{0}}}\left[\mathbb{G}_{M_{\nu}}\log\frac{p_{\theta_{0} j}}{p_{\theta j}}\right]^{2}\label{gbound}.
\end{align}
\end{subequations}
We proceed to bound the two terms in Equation~\ref{gbound}, from above.

\begin{align*}
&M_{\nu}\mathbb{E}_{P_{\theta_{0}},P_{\nu}}
\left(\sqrt{M_{\nu}}\left[\mathbb{P}^{\pi}_{M_{\nu}}-\mathbb{P}_{M_{\nu}}\right]\log\frac{p_{\theta_{0} j}}{p_{\theta j}}\right)^{2} \\
&=M_{\nu}\mathbb{E}_{P_{\theta_{0}},P_{\nu}}\left(\frac{1}{M_{\nu}}\mathop{\sum}_{i=1}^{M_{\nu}}\left(\frac{\delta_{\nu ji}}{\pi_{\nu ji}}-1\right)\log\frac{p_{\theta_{0} ji}}{p_{\theta ji}}\right)^2\\
&=\frac{1}{M_{\nu}}\mathop{\sum}_{i,\ell\in U_{\nu j}}\mathbb{E}_{P_{\theta_{0}},P_{\nu}}\left[\left(\frac{\delta_{\nu ji}}
{\pi_{\nu ji}}-1\right)\left(\frac{\delta_{\nu j\ell}}
{\pi_{\nu j\ell}}-1\right)\log\frac{p_{\theta_{0} ji}}{p_{\theta ji}}\log\frac{p_{\theta_{0} j\ell}}{p_{\theta j\ell}}\right]\\
&= \displaystyle\frac{1}{M_{\nu}}\mathop{\sum}_{i = \ell\in U_{\nu j}}\mathbb{E}_{P_{\theta_{0}}}\left[\mathbb{E}_{P_{\nu}}\left\{\frac{\delta_{\nu ji}}{\pi_{\nu ji^{2}}}-2\frac{\delta_{\nu ji}}{\pi_{\nu ji}}+1\middle\vert \mathcal{A}_{\nu j}\right\}\left(\log\frac{p_{\theta_{0} ji}}{p_{\theta ji}}\right)^{2}\right]\\
&+ \displaystyle\frac{1}{M_{\nu}}\mathop{\sum}_{i \neq \ell\in U_{\nu j}}\mathbb{E}_{P_{\theta_{0}}}\left[\mathbb{E}_{P_{\nu}}\left\{\frac{\delta_{\nu ji}\delta_{\nu j\ell}}{\pi_{\nu ji}\pi_{\nu \ell}}-\frac{\delta_{\nu ji}}{\pi_{\nu ji}}-\frac{\delta_{\nu j\ell}}{\pi_{\nu j\ell}}+1\middle\vert \mathcal{A}_{\nu j}\right\}\log\frac{p_{\theta_{0} ji}}{p_{\theta ji}}\log\frac{p_{\theta_{0} j\ell}}{p_{\theta j\ell}}\right]\\
&= \displaystyle\frac{1}{M_{\nu}}\mathop{\sum}_{i = \ell\in U_{\nu j}}\mathbb{E}_{P_{\theta_{0}}}\left[\left(\frac{1}{\pi_{\nu ji}}-1\right)\left(\log\frac{p_{\theta_{0} ji}}{p_{\theta ji}}\right)^{2}\right]\\
&+ \displaystyle\frac{1}{M_{\nu}}\mathop{\sum}_{i \neq \ell\in U_{\nu j}}\mathbb{E}_{P_{\theta_{0}}}\left[\left(\frac{\pi_{\nu ji\ell}}{\pi_{\nu ji}\pi_{\nu j\ell}}-1\right)\log\frac{p_{\theta_{0} ji}}{p_{\theta ji}}\log\frac{p_{\theta_{0} j\ell}}{p_{\theta j\ell}}\right]\\
&\displaystyle\leq 4\kappa^{-2}\epsilon_{M_{\nu}}^{2}\mathop{\sup}_{\nu}\left[\frac{1}{\mathop{\min}_{i\in U_{\nu j}}\pi_{\nu ji}}\right] +
4\kappa^{-2}\epsilon_{M_{\nu}}^{2}\left(M_{\nu}-1\right)\mathop{\sup}_{\nu}\mathop{\max}_{i\neq\ell\in U_{\nu j}}\left[\middle\vert \frac{\pi_{\nu ji\ell}}{\pi_{\nu ji}\pi_{\nu j\ell}} - 1\middle\vert\right]\\
&\leq4\kappa^{-2}\epsilon_{M_{\nu}}^{2}\left(c_{3} + \gamma\right),
\end{align*}
where we have applied assumptions (A8) and (A7) for the second and third terms in the last inequality.  We additionally note that $\pi_{\nu ji\ell} = \pi_{\nu \ell}$ when $i = \ell,~i,\ell \in U_{\nu j}$.
Through successive conditioning and bounding we have separated out from the joint expectation with respect to the population generating distribution, $P_{\theta_{0}}$, and the sampling design distribution, $P_{\nu}$, an expectation with respect to only $P_{\theta_{0}}$.  Since $\left\{\log\frac{p_{\theta ji}}{p_{\theta_{0} ji}}\right\}_{i = 1,\ldots,M_{\nu}}$ are independent, we may employ Bernstein's inequality to achieve the bound for the second moment of $\log\frac{p_{\theta ji}}{p_{\theta_{0} ji}}$, which produces the first term of the last equation, as follows:
\begin{equation}
\mathbb{E}_{P_{\theta_{0}}}\left(\log\frac{p_{\theta_{0} ji}}{p_{\theta ji}}\right)^{2}
 \leq 2!\kappa^{-2}\mathbb{E}_{P_{\theta_{0}}}\left[\exp\left(\kappa\log\frac{p_{\theta_{0} ji}}{p_{\theta ji}}
 \right)-1\right]\leq 4\kappa^{-2}\epsilon^{2}_{M_{\nu}},
\end{equation}
where we used the inequality $(\kappa x)^{2}/2! \leq e^{\kappa x} -1$ for $x \geq 0$, Bernstein's inequality and (A4).

The expectation of the centered and scaled empirical process (taken with respect to the population generating distribution) in the second additive term of Equation~\ref{gbound} is trivially bounded from above by,
\begin{equation*}
\mathbb{E}_{P_{\theta_{0}}}\left[\mathbb{G}_{M_{\nu}}\log\frac{p_{\theta_{0} j}}{p_{\theta j}}\right]^{2} \leq \mathop{\sup}_{\nu}\mathop{\max}_{i \in U_{\nu j}}\mathbb{E}_{P_{\theta_{0}}}\left(\log\frac{p_{\theta_{0}ji}}{p_{\theta ji}}\right)^{2} \leq 4\kappa^{-2}\epsilon^{2}_{M_{\nu}}
\end{equation*}

Finally,
\begin{align}
&\mbox{Pr}_{\theta_{0},P_{\nu}}\left\{\mathop{\int}_{\Theta_{\epsilon_{M_{\nu}}}}\displaystyle\mathop{\prod}_{i=1}^{N_{\nu}}\frac{p^{\pi}_{\theta ji}}{p^{\pi}_{\theta_{0} ji}}\Pi\left(d\Theta\right)\leq \exp\left[-r_{1}N_{\nu}t\right]\right\}\nonumber\\
&\leq \frac{4\kappa^{-2}\epsilon^{2}_{M_{\nu}}\left[c_{3}+1+\gamma\right]}{\left(r_{1}g_{1}\sqrt{M_{\nu}}\gamma t - \left(c_\pi g_{2}\gamma
-\left(\kappa\gamma\right)^{-1}\right)\sqrt{M_{\nu}}\gamma\epsilon^{2}_{M_{\nu}}\right)^{2}}\\
&\mathop{\leq}^{(i)}\frac{4\kappa^{-2}\epsilon^{2}_{M_{\nu}}\left[c_{3}+1+\gamma\right]}{h^{2}}\\
&\mathop{\leq}^{(i)} \frac{4\kappa^{-2}\epsilon^{2}_{M_{\nu}}\left[c_{3}+1+\gamma\right]}{h^{2}}\\
&\mathop{\leq}^{(ii)} \frac{\frac{2h\left[c_{3}+1+\gamma\right]}{\kappa\sqrt{M_{\nu}}}}{h^{2}}\\
&\leq \frac{2\left[c_{3}+1+\gamma\right]}{\kappa h\sqrt{M_{\nu}}} \\
&\mathop{\leq}^{(iv)} \frac{\left[c_{3}+1+\gamma\right]}{M_{\nu}t}\\
&\leq \frac{1}{r_{2}M_{\nu}t}
\end{align}
where $(i)$ follows from setting $h = r_{1}g_{1}\sqrt{M_{\nu}}\gamma t - \left(c_\pi g_{2}\gamma
-\left(\kappa\gamma\right)^{-1}\right)\sqrt{M_{\nu}}\gamma\epsilon^{2}_{M_{\nu}}$. Next, $(ii)$ follows by bounding $r_{1}$ from below with $r_{1} \geq \left(c_{\pi}g_{2} +3\left(\kappa\gamma\right)^{-1}\right)/g_{1}$.  Plugging the bound for $r_{1}$ into $h$ and replacing
$\epsilon^{2}_{M_{\nu}} < t$ with $t$ results in $h \geq\left[\left(c_\pi g_{2} -3\left(\kappa\gamma\right)^{-1}\right) -
\left(c_\pi g_{2} -\left(\kappa\gamma\right)^{-1}\right)\right]\sqrt{M_{\nu}}\gamma t = 2\kappa^{-1}\sqrt{M_{\nu}} t \leq
2\kappa^{-1}\sqrt{M_{\nu}}\epsilon^{2}_{M_{\nu}}$.  Re-arrange and achieve $4\kappa^{-2}\epsilon^{2}_{M_{\nu}} \leq 2h/\kappa\sqrt{M_{\nu}}$.  Continuing, $(iv)$ is achieved by
further algebra to $\kappa h \geq 2\sqrt{M_{\nu}} t$.

Finally, we set $r_{2} = 1/\left[c_{3\nu}+1+\gamma\right] \leq 1$, which is a function of the sampling design; in particular, $r_{2}$ is largest, which produces the fastest rate of decrease in the bound for the probability, when the sampling design is characterized by nearly independent samples and the gradient of the weights is $1$.

This concludes the proof.
\end{proof}

\section{Hierarchical Model for Current Employment Statistics Survey Data}

The specification of our probability model is completed by specifying the following priors,
\begin{subequations}
\label{dpmix}
\begin{align}
\mathop{\Theta}^{Q\times T} &\sim \mathcal{N}_{Q \times T}({0},\mathop{P_{2}^{-1}}^{Q\times Q} \circ P_{3}^{-1})\label{priorTheta}\\
\mathop{\Gamma_{\ell}}^{Q\times T} &\iid \mathcal{N}_{Q \times T}({0},P_{8}^{-1}\circ P_{6}^{-1}),~\ell = 1,\ldots,L\label{priorGamma}\\
P_{r} &\sim \mathcal{W}_{\tiny\mbox{dim}\{r\}}\left(\nu + \mbox{dim}\{r\} - 1, Q_{r}\right)\label{wish}\\
Q_{r} &= 2\nu\mbox{diag}(a_{r1},\ldots,a_{r \tiny\mbox{dim}\{r\}});~ r = 2, 8 \label{folded}\\
a_{r1},\ldots,a_{r \tiny\mbox{dim}\{r\}} &\iid \mathcal{G}(1/2,1)\\
P_{s} &= D-\rho_{s}\Omega;~ \rho_{s} \sim \mathcal{U}(0,1);~ s = 3, 6\\
\tau_{q}^{-1/2} &\iid \mathcal{C}(0,1)
\end{align}
\end{subequations}
The constructions for $\mathop{\Theta}$ and $\{\mathop{\Gamma_{\ell}}\}$ (${\ell = 1,\ldots,L}$) employ separable or tensor product formulations \citep{hoff2011} for precision matrices, where each precision matrix permits the discovery of correlations among the $Q = 2$ response variables and among the $T = 12$ months. The prior formulation specified in \eqref{wish} and \eqref{folded} generalize the Wishart prior by constructing the mean, $Q_{r}$, as a diagonal matrix parameterized by $a_{r1},\ldots,a_{r \tiny\mbox{dim}\{r\}}$, where $\tiny\mbox{dim}\{r\}$ denotes the dimension of $P_{r}$.  These parameters, in turn, receive Gamma priors to de-couple the correlations between variances and correlations present under a Wishart prior.  In particular, this prior induces marginally folded-t distributions with $\nu$ degrees of freedom on the standard deviations  and marginally uniform distributions on the correlations when $\nu = 2$ \citep{huang2013}.  We select this more flexible prior because a primary focus in our modeling is to borrow strength over response variables, industries, and months.

Precision matrices, $\left(\mathop{P_{3}},\mathop{P_{6}}\right)$, are constructed as proper conditional autoregressive formulations, where $\rho_{3}~ (\rho_{6})$ may be interpreted as a strength-of-temporal-association.   $\Omega = \{\omega_{ij}\}$ is a $T\times T$ adjacency matrix where $\omega_{ij} = 1$ if months $i$ and $j$ are adjacent; else, $\omega_{ij} = 0$. $D$ is a $T\times T$ diagonal matrix of row sums of $\Omega$ such that the precisions for months with a larger number of neighbors will be higher than those with a relatively smaller number of neighbors.

We now illustrate pseudo posterior computations for $\mathop{\Theta}$ and $\{\mathop{\Gamma_{\ell}}\}$.  We jointly sample $Q\times T$, $\mathop{\Theta}$, in one step under an elliptical slice sampler \citep{murray2010} since the underlying posterior is non-conjugate. We draw $\mathop{\Theta}$ from its prior in \eqref{priorTheta} and form a convex combination with the previously sampled value that is parameterized to lie on an ellipse.  The proposal is evaluated with the log-pseudo likelihood,
\begin{equation}
\begin{split}
&\log~L^{\pi}\left(\Theta\vert {Y},\{\Gamma_{\ell}\},\{\tau_{q}\}\right) \propto \\ &\mathop{\sum}_{q=1}^{Q}\mathop{\sum}_{c=1}^{n_{c}}\tilde{w}_{i\{c\}}\left[-\left(\tau_{q} + y_{cq}\right)\log\left(\tau_{q}+ C_{1,cq}\exp\left(\theta_{q~t\{c\}}\right)\right) + y_{cq}\theta_{q~t\{c\}}\right]\label{logliketheta},
\end{split}
\end{equation}\label{likeTheta}
where $C_{1,cq} = \exp\left(\gamma_{q~ t\{c\}\ell\{c\}}z_{c}\right)$ is independent of $\mathop{\Theta}$ and \eqref{logliketheta} is the sampling-weighted kernel of the negative binomial log likelihood after dropping all additive terms independent of $\Theta$.

We similarly jointly sample each $Q\times T$, and $\Gamma_{\ell}$ ($\ell = 1,\ldots,L$), using the elliptical slice sampler sampler with a proposal formed with a convex combination of a prior draw from \eqref{priorGamma} and the last sampled value that is subsequently evaluated with,
\begin{equation}
\begin{split}
  &\log~L^{\pi}\left(\Gamma_{\ell}\vert {Y},\Theta,\{\tau_{q}\}\right) \propto \\ &\mathop{\sum}_{q=1}^{Q}\mathop{\sum}_{c=1}^{n_{c}}\tilde{w}_{i\{c\}}\left[-\left(\tau_{q} + y_{cq}\right)\log\left(\tau_{q}+ C_{2,cq}\exp\left(\gamma_{q~ t\{c\}\ell\{c\}}z_{c}\right)\right) + y_{cq}\gamma_{q~ t\{c\}\ell\{c\}}z_{c}\right],
\end{split}
\end{equation}\label{likeGamma}
where $C_{2,cq} = \exp\left(\theta_{q~t\{c\}}\right)$ is independent of $\Gamma$.

The over-dispersion parameters, $\{\tau_{q}\}$, are sampled in a slice sampler \citep{Neal00slicesampling} from the following pseudo posterior,
\begin{equation}
\begin{split}
&\log~\pi^{\pi}\left(\tau_{q}\vert {Y},\Theta,\{\Gamma_{\ell}\}\right) \propto \\
&\left[\mathop{\sum}_{c=1}^{n_{c}}\tilde{w}_{i\{c\}}\right]\left[\tau_{q}\log~\tau_{q} - \log\Gamma\left(\tau_{q}\right)\right] \\
&+ \mathop{\sum}_{c=1}^{n_{c}}\tilde{w}_{i\{c\}}\left[-\left(\tau_{q}+ y_{cq}\right)\log\left(\tau_{q}+\exp\left(\psi_{cq}\right)\right)+\log\Gamma\left(\tau_{q}+y_{cq}\right)\right]\\
& -\frac{1}{2}\log~\tau_{q} - \log(1+\tau_{q}),
\end{split}
\end{equation}
where $\Gamma\left(\cdot\right)$ is the Gamma function.

These constructions for the full sample are readily purposed to estimation on the subsets under our stochastic approximation of (6) by normalizing the set of $m$ unit weights for subset $j$, $\displaystyle\mathop{\sum}_{i\in S_{j}}\tilde{w}_{ji} = n$, the observed full data sample size.  The remaining precision parameters are sampled in the usual way, conditional on $\mathop{\Theta}$ and $\{\mathop{\Gamma_{\ell}}\}_{\ell = 1,\ldots,L}$, with no application of sampling weights.

\begin{table}[h]
\def~{\hphantom{0}}
\caption{\it List of $23$ $2$- digit supersectors and the distribution of the number of establishments in the Current Employment Statistics Survey data.}
{\footnotesize
\begin{tabular}{rlr}
  \hline
 & Supersector & Number of sampled units \\
  \hline
1 & Retail Trade (44) & 6470 \\
  2 & Accommodation and Food Services & 6282 \\
  3 & Finance and Insurance & 3489 \\
  4 & Health Care and Social Assistance & 3393 \\
  5 & Professional, Scientific, and Technical Services & 2599 \\
  6 & Retail Trade (45) & 1888 \\
  7 & Other Services (except Public Administration) & 1772 \\
  8 & Construction & 1440 \\
  9 & Information & 1361 \\
  10 & Administrative and Support and Waste Management and Remediation Services & 1359 \\
  11 & Wholesale Trade & 1167 \\
  12 & Real Estate and Rental and Leasing & 963 \\
  13 & Manufacturing (33) & 847 \\
  14 & Transportation and Warehousing (48) & 521 \\
  15 & Management of Companies and Enterprises & 519 \\
  16 & Manufacturing (32) & 474 \\
  17 & Arts, Entertainment, and Recreation & 458 \\
  18 & Transportation and Warehousing (49) & 451 \\
  19 & Educational Services & 424 \\
  20 & Manufacturing (31) & 355 \\
  21 & Utilities &  86 \\
  22 & Mining, Quarrying, and Oil and Gas Extraction &  51 \\
  23 & Agriculture, Forestry, Fishing and Hunting &  21 \\
   \hline
\end{tabular}
}
\label{tab:ss}
\end{table}

 \end{document}